\documentclass[graybox]{svmult}

\usepackage{type1cm}        
\usepackage{makeidx}         
\usepackage{graphicx}        
\usepackage{multicol}        
\usepackage[bottom]{footmisc}

\usepackage{newtxtext}       %
\usepackage{newtxmath}       
\usepackage{braket}

\usepackage[firstabbrev]{authorindex}
\usepackage{url}
\usepackage[square, numbers]{natbib}

\usepackage{csquotes}
\usepackage{comment}
\usepackage[normalem]{ulem}

\usepackage{amsmath}

\makeindex             
\usepackage{bm,bbm}       %
\newcommand{\E}{\mathbb{E}}

\begin{document}

\title*{Emulating the First Principles of Matter: A Probabilistic Roadmap}

\author{Jianzhong Wu and Mengyang Gu}
\institute{
Jianzhong Wu \at Department of Chemical and Environmental Engineering, University of California, Riverside, CA 92521, USA \email{jwu@engr.ucr.edu}
\and 
Mengyang Gu \at Department of Statistics and Applied Probability, University of California, Santa Barbara, CA 93106, USA \email{mengyang@pstat.ucsb.edu}
}
%
%
\maketitle

\abstract{This chapter provides a tutorial overview of first principles methods to describe the properties of matter at the ground state or equilibrium. It begins with a brief introduction to quantum and statistical mechanics for predicting the electronic structure and diverse static properties of of many-particle systems useful for practical applications. Pedagogical examples are given to illustrate the basic concepts and simple applications of quantum Monte Carlo and density functional theory \textemdash  two representative methods commonly used in the literature of first principles modeling. In addition, this chapter highlights the practical needs for the integration of physics-based modeling and data-science approaches to reduce the computational cost and expand the scope of applicability. A special emphasis is placed on recent developments of statistical surrogate models to emulate first principles calculation from a probabilistic point of view. The probabilistic approach provides an internal assessment of the approximation accuracy of emulation that quantifies the uncertainty in predictions.   Various recent advances toward this direction establish a new marriage between Gaussian processes and first principles calculation, with physical properties, such as translational, rotational, and permutation symmetry, naturally encoded in new kernel functions. Finally, it concludes with some prospects on future advances in the field toward faster yet more accurate computation leveraging a synergetic combination {of} novel theoretical concepts and efficient numerical algorithms.}
 

\section{First principles modeling}
\label{sec:First}
Mathematically speaking, an essential task to predict the properties of matter from first principles is by solving the Schr\"{o}dinger equation. While the task is rather straightforward if one is concerned only with the properties of non-interacting particles (such as ideal Fermions or the lone electron in a hydrogen atom), the problem rapidly becomes much too complicated when the procedure is extended to non-ideal systems consisting of more than a single particle. By non-ideal we mean interactions between particles such as the Coulomb potential between charged species.  Not only is the dimensionality of the wave function linearly increasing with the number of particles in the system, but additional considerations must also be taken to account for correlation effects due to particle-particle interactions. The latter is responsible for the non-random spatial arrangement of particles in a many-body system which gives rise to the system symmetry and, for macroscopic systems, the rich phase behavior of matter in response to the changes of thermodynamic conditions. 

The fundamental principles to describe particle-particle interactions have been well established within the framework of quantum mechanics (QM). On the other hand, structure formation and phase transition in macroscopic systems are dictated by the fundamental laws of thermodynamics and can be described, at least in principle, by statistical mechanics (SM). From a practical perspective, the situation can thus be summarized, as famously stated many years ago by Paul M. Dirac \cite{RN22}, 
\begin{quotation} The underlying physical laws necessary for the mathematical theory of a large part of physics and the whole of chemistry are thus completely known, and the difficulty is only that the exact application of these laws leads to equations much too complicated to be soluble. It therefore becomes desirable that approximate practical methods of applying quantum mechanics should be developed, which can lead to an explanation of the main features of complex atomic systems without too much computation.\end{quotation}

Since the beginning of the last century, a perennial effort in the scientific community has been devoted to the development of analytical and numerical schemes to approximate the general procedures of QM/SM calculations such that they can be applied to materials and chemical systems to attain useful results that would satisfy the practical needs. Such efforts remain active today. The theoretical methods and their applications to diverse problems of practical interest constitute a major component of curriculum for a wide variety of disciplines in physical sciences and engineering. Numerous textbooks of QM and SM are readily available on both the fundamental principles and practical applications. Here we introduce only the essential mathematical procedures to describe the properties of many-body systems with a minimal exposure to the physical details. The emphasis is placed on a few theoretical approaches commonly used in the literature to predict the electronic structure and macroscopic properties of matter. To establish connections with problems of practical interests, we will elucidate how the electronic structure is related to various physical and chemical properties of chemical systems and materials.     
  
 \subsection{Quantum completeness}
\label{subsec:QC}        
The ultimate goal of first principles modeling is to predict the properties of matter based on its fundamental ingredients, \textit{i.e.}, electrons and nuclei as appeared in gas, liquid, solid or plasma \textemdash the four natural states of matter commonly observable in daily life. Electrons are elementary particles.  Each electron has a negative unit charge, $-1.602176634 \times 10^{-19}$ C. Nuclei are made of neutrons and protons. In chemical systems and materials, nuclei may be represented by point charges under most conditions. 
  
Quantum mechanics asserts that matter exits in discrete quantum states, \textit{i.e.}, a set of parameters to describe the ultimate details of the system. The properties of matter are manifested as the expectation of the collective behavior of the underlying particles in different quantum states. Once the quantum states are identified, we can in principle determine all properties of the system.  

To elucidate the essential mathematical procedure, consider in general a system containing $N_{e}$ electrons and $N_\alpha$ nuclei of type $\alpha$. At the macroscopic scale, all natural states of matter satisfy the condition of charge neutrality, \textit{i.e.}, the total electron charge is exactly balanced by that of the nuclei. Therefore, the condition of charge neutrality requires 
\begin{equation}
N_{e} - \sum_{\alpha} Z_\alpha N_\alpha = 0 
\end{equation}
where $Z_\alpha$ is a positive integer standing for the valency of the nuclear charge. This integer coincides with the atomic number for nuclear particle $\alpha$.     
 
At any moment, the electrons and nuclei may exist in a multitude of quantum states satisfying the Schr\"{o}dinger equation:   
\begin{equation}
\label{eq:SE}
\hat H \ket \Psi = E  \ket \Psi
\end{equation}
where $\ket {\Psi}$ represents a quantum state as specified by wave function $\Psi$, $E$ is is a scalar value representing the system energy, and $\hat H$ denotes the Hamiltonian of the system. The dynamic properties of the system can be described by the time-dependent Schr\"{o}dinger equation, which is not of concern in this work. 

In quantum mechanics, Hamiltonian is a mathematical operator defining the kinetic and potential energies of the system, \textit{i.e.}, the energies affiliated with the motions of individual particles and particle-particle interactions. For a system containing $N_{e}$ electrons and $N_\alpha$ nuclei of type $\alpha$, the Hamiltonian is given by
\begin{equation}
\label{eq:H}
\hat H =  - \sum_{i}^{N} \frac{\hbar ^2}{2{m_i}}\nabla _i^2  + \frac{e^2}{8\pi \varepsilon_0}
\sum_{i}^{N} \sum_{j \neq i}^{N} \frac{Z_{i} Z_{j}} {\left| {\bf{r}}_{i} - {\bf{r}}_{j} \right|}\
\end{equation}
where $\nabla_i$ is an Laplacian operator of the $i$th particle, $i=1,...,N$,  with $N \equiv \sum_{\alpha}N_{\alpha}+N_{e}$ being the total number of particles in the system, $m_i$ stands for the rest mass of particle $i$, ${\bf{r}}_{i}$ is the particle position, $\left| {\bf{r}}_{i} - {\bf{r}}_{j} \right|$ represents the Euclidean distance between particles $i$ and $j$, $\hbar=h/2\pi$ is the reduced Planck constant, $e$ is the unit charge, and ${\varepsilon _0}$ is the free-space permittivity. The first term on the right defines the total kinetic energy, which is affiliated with momenta of all particles in the system. The second term prescribes the potential energy, arising from the electrostatic interaction between electrons and nuclei. The electrostatic potential has an expression formally identical to that given by the Coulomb's law for classical particles. 

The wave function has the units of one over square root of volume. It can be represented in terms of the system configuration, \textit{i.e.}, a set of coordinates that define the positions and angular momenta of individual particles, ${\bf {x}}^{N} \equiv \{ {\bf{x}}_{1}, {\bf{x}}_{2},..., {\bf{x}}_{N} \}$. Here each vector ${\bf{x}}_{i} \equiv \{ {\bf{r}}_{i}, s_{i} \}$ specifies the position ${\bf{r}}_{i}$ and spin state ${\bf{s}}_{i}$ of particle $i$. The electron spin is affiliated with its intrinsic angular momentum as that for an elementary particle; it takes only two possible values that are conventionally denoted as $\ket \uparrow$ and $\ket \downarrow$, or simply spin up and down states. By contrast, the nuclear spin arises from its subatomic constituents, \textit{i.e.},  protons and neutrons. The nuclear spin is commonly treated as a single entity, which is invariant with the quantum states of the system. Therefore, we may describe the configuration of a system containing $N_{e}$ electrons and $N_{n}=\sum_\alpha N_\alpha$ nuclei using  $4N_{e} + 3N_{n}$ variables. As each variable represents one degree of freedom, the wave function $\Psi$ has the dimensionality of $4N_{e} + 3N_{n}$.      

Mathematically, Eq.\eqref {eq:SE} represents an eigenvalue problem. The energy levels and wave functions are related to the eigenvalues and eigenfunctions corresponding to operator $\hat H$. The eigenstates are also known as the pure states, whose wave functions satisfy the orthonormality condition  
\begin{equation}
\label{eq:ON}
\int {d{\bf{x}}^{N} } \Psi^{*}_{n} ({\bf {x}}^{N})\Psi_{m} ({\bf {x}}^{N})  = \left \{ \begin{array}{lr}
    1 &  n = m\\
    0 &  n \neq m
  \end{array}
\right.
\end{equation}
where superscript * represents complex conjugate, integers $n$ and $m$ are quantum numbers. At each pure state $n$, $\left| \Psi_{n} ({\bf {x}}^{N})\right|^2$ represents the probability density of the system in configuration ${\bf {x}}^{N}$. 

For an arbitrary quantum state, the wave function can be expressed as a supposition of pure states 
\begin{equation}
\Psi ({\bf {x}}^{N}) = \sum_{n} \alpha_{n} \Psi_{n} ({\bf {x}}^{N})  
\end{equation}
where subscript $n$ denotes an eigenstate, and 
\begin{equation}
\alpha_{n} = \int d{\bf {x}}^{N}  \Psi^{*}_{n}({\bf {x}}^{N}) \Psi({\bf {x}}^{N})
\end{equation}
A mixed quantum state is referred to as one that can be written as a linear combination of more than one pure states, \textit{i.e.}, $\alpha_{n} \neq 0$ for more than one pure states. 

For the system at a particular quantum state, any observable property can be evaluated from the multidimensional integrations 
\begin{equation}
\left<\hat{A}\right>_{\Psi}=  \frac{ \int d{\bf {x}}^{N} \Psi^{*}({\bf {x}}^{N}) \hat{A} \Psi({\bf {x}}^{N}) }{ \int d{\bf {x}}^{N} \Psi^{*}({\bf {x}}^{N}) \Psi({\bf {x}}^{N})} 
\end{equation}
where operator $\hat{A}$ denotes an observable quantity, and $\left< ... \right>_{\Psi}$ stands for quantum expectation, \textit{i.e.}, the expectation value of an observable property of the system in accordance with wave function ${\Psi}$. 

While the affiliation of particles with positions and spin states is intuitively appealing, one should keep in mind that, unlike classical particles, quantum particles are not allowed to have definite coordinates at any instance and thus, strictly speaking, cannot be \enquote{tagged} with specific positions and spin states. At any moment, quantum particles may assume positions corresponding to a superposition of all possible  pure states.

\begin{backgroundinformation}{A One-Particle Problem}

The one-particle problem is helpful to elucidate some basic concepts related to the Schr\"{o}dinger equation. If we consider a single particle in free space, the Schr\"{o}dinger equation would be reduced to  
\begin{equation}
\label{eq:1P}
-\frac{\hbar ^2}{2{m}}\nabla ^2 \Psi ({\bf {r}})  = E \Psi ({\bf{r}})
\end{equation}
The differential equation can be readily solved with the periodic boundary conditions (PBC) 
\begin{equation}
\label{eq:PB1}
\Psi ({\bf {r}})  = \Psi ({\bf {r+L}}) 
\end{equation}
where ${\bf {L}} \equiv (L,L,L)$, and $L>0$ represents the system size. The PBC may be understood as a division of free space into cubic boxes of side length $L$  such that each box contains an imaginary particle imaging the position of the real particle under consideration. We assume that the real and imaginary particles are assumed identical but do not interacting with each other. 

From Eqs.\eqref{eq:1P} and \eqref{eq:PB1}, we can easily find the wave function by using the Fourier transform:
\begin{equation}
\Psi ({\bf {r}})  = \frac{\exp(i \bf {r \cdot k})} {\sqrt {V}}
\end{equation}   
where $V=L^3$, and ${\bf {k}}$ is a 3-dimensional the vector given by
\begin{equation}
{\bf {k}} = \frac{2\pi} {L} {\bf {n}}
\end{equation}
with ${\bf {n}} = (n_x, n_y, n_z)$, $n_{x,y,z}=0, \pm1, \pm2,\pm 3, ....$ represents quantum numbers. 

It is straightforward to verify that the wave function satisfies the orthonormal conditions
\begin{equation}
\int d{\bf {r}} \left| \Psi_{\bf {k}} ({\bf {r}})\right|^2  = 1
\end{equation}  
and
\begin{equation}
\int d{\bf {r}} \Psi^{*}_{\bf {k}}({\bf {r}}) \Psi_{\bf {k'}}  ({\bf {r}})  = \delta_{\bf {k,k'}}
\end{equation}  
where $\delta_{\bf {k,k'}}$ denotes the Kronecker delta function, which is equal to 1 if $\bf {k=k'}$ and zero otherwise. At each quantum state, the particle is uniformly distributed inside the box, \textit{i.e.}, the probability density of finding a particle is everywhere uniform. 

At each quantum state, we can find the particle energy from the Schr\"{o}dinger equation:  
\begin{equation}
E = \frac{\hbar ^2 k^2}{2{m}} = \frac{h^2}{2{mL^2}} (n^2_x + n^2_y + n^2_z)
\end{equation}
At the ground state, the particle has a minimum energy of $E_0 = {h^2}/(2{mL^2})$, which has a degeneracy of 6 corresponding to all possible assignments of the quantum numbers leading to $n^2_x + n^2_y + n^2_z=1$. Unlike a classical particle, a quantum particle cannot have a zero energy as required by the uncertainty principle.   
 \end{backgroundinformation}

The Schr\"{o}dinger equation is applicable to systems with any number of particles, either finite or infinite. If the system is isolated from its surroundings, the total energy is fixed, and the number of quantum states corresponding to the particular energy is called degeneracy. In other words, degenerate quantum states have the same energy. The ground state is referred to as the state of a system when it has the minimum energy.  If the system allows to exchange energy with its surroundings (\textit{e.g.}, in contact with a thermal bath), the total energy fluctuates so that the system becomes accessible to different excited states. 

In statistical mechanics, the quantum states are also known as microstates. At each microstate, we know  the microscopic details of the system such as energy and particle positions. For a system with a given number of particles, volume and temperature, the probability of different microstates is determined by the Boltzmann distribution
\begin{equation}
p_{\nu} = \frac{\exp(-E_\nu/k_{B}T)}{Q}  
\end{equation}
where $\nu$ denotes a microstate, $k_{B}$ is the Boltzmann constant, $T$ is the absolute temperature, and $Q \equiv \sum_{\nu} \exp(-E_\nu/k_{B}T)$ is called the canonical partition function. Accordingly, the average energy of the system is given by
\begin{equation}
\left<E \right> \equiv  \sum_{\nu} p_{\nu} E_{\nu}
\end{equation}
where $\left< ... \right>$ stands for the ensemble average. From the partition function, we can derive in principle all thermodynamic properties \cite{RN25}.

Typically, a molecule contains no more than a few types of nuclei. A similar statement can be made for most materials.  However, most systems of practical concern contain a large number of particles. For a macroscopic system, the total number of particles, $N=N_{e} + \sum_{\alpha} N_\alpha$, is astronomically large ($\sim10^{23}$) and approaches infinity in the thermodynamic limit. Because the dimensionality of wave function scales linearly with the total number of particles, the Schr\"{o}dinger equation becomes \enquote{\textit {much too complicated to be soluble}} as the number of particle increases. For practical applications, the essential task is thus to develop \enquote{\textit {approximate methods of applying quantum and statistical mechanics}}.   

\subsection{Born-Oppenheimer approximation}
\label{subsec:BO} 
The Born-Oppenheimer (BO) approximation assumes that the electron degrees of freedom can be decoupled from those corresponding to the nuclei, and that the latter can be represented classical particles with negligible size. The assumption is justifiable because a nuclear particle occupies little volume inside each atom. Besides, the electron rest mass $m_{e}$ is much smaller than that of a proton $m_{p}$, the smallest nuclear particle ($m_{p}/m_{e} \approx 1836$). The huge difference in rest mass implies that the electron motion is faster than that of nuclei by several orders of magnitude. As a result, electrons may be considered to be in the ground state at any configuration of the nuclei. With each nucleus subject to a force owing to interaction with other nuclei and the inhomogeneous electron distributions, the nuclear motion follows the classical laws of physics that can be integrated with numerical procedures (\textit{viz.}, molecular dynamics  simulation). 

As the degree of freedom related to nuclear spins is irrelevant for most chemical systems, the configuration of nuclei can then be specified in terms of their positions, ${\bf {R}}_{N_n} \equiv ({\bf {R}}_1, {\bf {R}}_2,...,{\bf {R}}_{N_n})$. At a time scale sufficiently long for electron relaxation but short for the motion of nuclei, which is on the oder of a fraction of femtosecond ($10^{-15}$s), electrons are approximately in a stationary state subject to an external field arising from electrostatic interactions with the nuclei
\begin{equation}
\label{eq:eex}
v({\bf {r}}) =  -\frac {e^2} {4\pi \varepsilon_0} \sum^{N_{n}}_{I=1} \frac {Z_{I} } {\left| {\bf{R}}_{I} - \bf{r} \right|}
\end{equation}    
The ground state energy and the electronic structure can be determined by solving the Schr\"{o}dinger equation
 \begin{equation}
\left( - \sum_{i}^{N} \frac{\hbar ^2}{2{m_e}}\nabla _i^2  + \frac{e^2}{4\pi \varepsilon_0}
\sum^N_{i} \sum_{j > i}^{N} \frac{1} {\left| {\bf{r}}_{i} - {\bf{r}}_{j} \right|}\ + v({\bf {r}}) \right)  \Psi ({\bf{x}}^{N})= E  \Psi ({\bf{x}}^{N})
\end{equation}
For simplicity of notation, from now on we replace $N_{e}$ with $N$, ${\bf {x}}^{N} \equiv \{ {\bf{x}}_{1}, {\bf{x}}_{2},..., {\bf{x}}_{N} \}$ represents the electron configuration,$\Psi ({\bf{x}}^{N})$ is the electron wave function, and $E$ is the total electronic energy.
          
Once the electron wave function is determined from the Schr\"{o}dinger equation, the force on each nucleus due to the inhomogeneous distribution of electrons can be calculated from the Hellmann-Feynman (HF) equation
\begin{equation}
\label{eq:HF}
{\bf {F_{I}}} =  \frac {Z_{I}e^2} {4\pi \varepsilon_0} \int d{{\bf {r}}} \hat\rho({\bf {r}}) \frac {{\bf {r}}-{{\bf {R}}_{I}}} {\left| {\bf{R}}_{I} - {\bf{r}} \right|^{3}}
\end{equation}  
where $\hat\rho({\bf {r}})$ stands for the electron density. The latter is related to the wave function
\begin{equation}
\label{eq:rho0}
\hat\rho({\bf {r}}) = \sum^{N}_{i=1} \int d{\bf {x}}^{N} \left| \Psi ({\bf{x}}^{N}) \right|^{2} \delta ({\bf {r-r}}_{i})  
\end{equation}
where $\delta ({\bf {r-r}}_{i})$ is the Dirac delta function.The physical meaning of Eq.\eqref {eq:HF} is intuitive: the overall force on nucleus $I$ due to the electrons is simply equal to the integration of the local electron number density multiplied by the Coulomb force.           

\begin{backgroundinformation}{Hydrogen Atom}
\label{HAtom}
A normal hydrogen atom contains two particles, \textit {i.e.}, one electron and one proton. Following the BO approximation, we may consider a hydrogen atom as an electron orbiting around the proton with the electron wave function described by the one-particle Schr\"{o}dinger equation 
\begin{equation}
\label{eq:1H}
\left( -\frac{\hbar ^2}{2{m_{pe}}}\nabla ^2 -\frac {e^2} {4\pi \varepsilon_0 r} \right) \Psi ({\bf {r}})  = E \Psi ({\bf{r}})
\end{equation}
where $1/m_{pe} \equiv1/m_{e}+1/m_{p}, r={\left|\bf{r}\right|}$ is the radial distance. Here the proton is placed at the center of the coordinate system. For a single electron, the intrinsic magnetic momentum and thus the spin number play no role in determining the electronic properties of the system. 

With the boundary conditions $\Psi(r)=\Psi{'}(r)=0$ as $r \rightarrow \infty$, Eqs.\eqref{eq:1H} yields an analytical solution. In spherical coordinates, the wave function is given by
\begin{equation}
\label{eq:1Hw}
\Psi (r, \theta,\phi)  = \mathcal{N}_{n,l} x_{n}^{l} e^{-x_{n}/2} \mathcal{L}_{n,l} (x_{n}) Y^{m}_{l}(\theta, \phi), 
\end{equation}
and the corresponding energy is
\begin{equation}
\label{eq:1He}
E = - \frac{m_{pe}e^{4}}{32\pi^{2}\epsilon^{2}_{0}\hbar^{2}n^{2} }. 
\end{equation}
In atomic physics, $(n,l,m)$ are known as principal, azimuthal, and magnetic quantum numbers, respectively. These quantum numbers, take the integer values of $n=1,2,3,...$, $l=0,1,2,...,n-1$, and $m=-l,...,l$, and define the atomic orbitals that are commonly used as the basis functions for the wave functions of other atoms and molecular systems. 

In Eqs.\eqref{eq:1Hw} and \eqref{eq:1He}, $\mathcal{N}_{n,l}$ is a normalization constant for the radial component of the wave function
\begin{equation}
\mathcal{N}_{n,l} = \sqrt{ \left(\frac {2m_{pe}}{a_{0}m_{e}} \right)^3 \frac {(n-l-1)!}{2n[n(n+l)!]^{3}}}
\end{equation}
where $a_{0} =4 \pi \epsilon_{0} \hbar^{2} / (m_{e} e^{2})$ is known as the Bohr radius. The universal constant,  $a_{0} =5.2917721... \times 10^{-11} $m is often used as the unit length. $\mathcal{L}_{n,l} (x)$ stands for an associated Laguerre polynomial of degree $(n-l-1)$ and order $(2l+1)$, $x_{n}=2rm_{pe}/(na_{0}m_{e})$ is the dimensionless radial distance, and $Y^{m}_{l}(\theta, \phi)$ is a spherical harmonic function of degree $l$ and order $m$. 

A stable hydrogen atom exists in the ground state. In this case, the quantum numbers are $n=1, m=0$, and $l=0$, and the minimum energy is 
\begin{equation}
\label{eq:E0}
E_{0} = - \frac{m_{pe} E_{Ryd} } {m_{e}} \approx -E_{Ryd} 
\end{equation}
where $E_{Ryd} = {e^{2}} /({8\pi \epsilon_{0} a_{0}})$ is known as the Rydberg energy. The Rydberg energy, $2.17872... \times 10^{-18} $ J, is a universal constant that is often used as a unit energy in atomic physics.\footnote{An alternative energy unit is hartree, 1 hartree = 2 rydberg = 27.211 eV }. Intuitively, Eq. \eqref{eq:E0} may be understood as the electrostatic energy between the electron and the proton at an average distance twice the Bohr radius.  

At the ground state, the wave function for a hydrogen atom  is given by 
\begin{equation}
\Psi (r) =  \frac {1}{\sqrt{(\pi a^{3}_{0})}} e^{-r/a_{0}}  
\end{equation} 
Correspondingly, the electron density is 
\begin{equation}
\rho(r) = \frac {e^{-2r/a_{0}}}{\pi a^{3}_{0}}   
\end{equation} 
The spherically symmetric function decays exponentially and has a maximum value of $1/(\pi a_{0}^{3} )$ at the nucleus (at $r = 0$).   

The ground-state energy, $E_{0} \approx$-13.598 eV, represents the energy change when an electron and a proton bind to form a stable hydrogen atom. This energy corresponds to the negative of the hydrogen ionization energy. The changes among different energy levels of the hydrogen atom explain its light emission spectrum, which represents a major triumph in the early development of quantum mechanics.     
 \end{backgroundinformation}
 
As illustrated in Box \ref{HAtom}, one of the simplest examples for the application of the BO approximation is provided by the first principle predictions for the spectrum and ionization energy of atomic hydrogen. In principle, a similar approach can be applied to polyatomic molecules by representing the molecular energy in terms of the electronic contribution plus those related to nuclear motions within the molecule, such as bond stretching, bond vibration, and molecular rotations. The BO model provides a theoretical basis to predict molecular spectroscopy and the thermodynamic properties of ideal gas systems. For a hydrogen atom, we fix the nuclear position which is treated as the center of coordinates for solving the Schr\"{o}dinger equation. When a system contains multiple nuclei, the electron distribution is in general anisotropic, leading to an atomic force on each nucleus responsible for the molecular configurations as well as atomic motions including chemical reactions. If the nuclei are treated as classical particles, we may describe the motions of nuclei using Newton's equations. The combination of quantum mechanics for the electronic structure calculations and classical physics for the nuclear motions constitutes the essential ideas of the Born-Oppenheimer molecular dynamics (BOMD) simulation. 

With the nuclei treated as classical particles, the BO approximation greatly simplifies the computational task to predict the properties of matter from first principles. Not only does the BO approximation reduce the dimensionally of the wave function, it also essentially transforms the complex quantum-mechanic problem to one that is only concerned with electronic structure calculations. Whereas the electronic wave function remains a multidimensional quantity, it represents the property of only a single component system. In particular, the electron density can be fully determined from the one-body external potential, a three-dimensional function that depends only on the nuclear positions (see Eq.\eqref {eq:eex}).     
 
\subsection{Quantum Monte Carlo simulation}
\label{subsec:QMC}  
Monte Carlo methods for solving the many-body Schr\"{o}dinger equation were suggested first by Metropolis and Ulam in 1949 \cite{RN30}. However, major breakthroughs were made not until the publication of a landmark work by Ceperley and Alder in 1980 \cite{RN27}. Today quantum Monte Carlo (QMC) simulation represents properly the most generic way to accurately predict electronic properties \cite{RN35, RN37}. 

The central idea of Monte Carlo methods is to generate a large number of samples using a stochastic process. It converts multidimensional operations in terms of simple mean-value evaluations \cite{RN22}. The statistical approach finds broad applications in various branches of mathematics for solving high-dimensional optimization problems and integro-differential equations. The development of the Metroplis (a.ka., the $MR^{2}T^{2}$ algorithm) marks a milestone for the broad use Monte Carlo methods in physical sciences. As stated befittingly in the introductory sentence of their famous paper \cite{RN30}, Monte Carlo methods are \begin{quotation} suitable for fast electronic computing machines, of calculating the properties of any substance ....\end{quotation}

The variational quantum Monte Carlo (VMC) represents one of the simplest ways to evaluate many-body electronic wave function by using Monte Carlo simulation. The basic idea is that the ground state energy satisfies the variational principle
\begin{equation}
\label{eq:VP}
E_{v} = \frac{ \int d{\bf {x}}^{N} \Psi^{*}({\bf {x}}^{N}) \hat{H} \Psi({\bf {x}}^{N}) }{ \int d{\bf {x}}^{N} \Psi^{*}({\bf {x}}^{N}) \Psi({\bf {x}}^{N})}  \geq E_{0} 
\end{equation}
where  $\Psi({\bf {x}}^{N})$ stands for the wave function of the system in an arbitrary quantum state. The inequality is rather intuitive because, by definition,  electrons in an arbitrarily quantum state must have an energy no less than the ground-state value. While the mathematic proof is also elementary, evaluation of the energy entails multidimensional integrations that cannot be performed with conventional numerical methods. 

In VMC, the multidimensional integration for the system energy is expressed in terms of an expectation value
\begin{equation}
E_{v} = \int d{\bf {x}}^{N} p({\bf {x}}^{N})E({\bf {x}}^{N})
\end{equation}    
where $E({\bf {x}}^{N})$ represents a local energy density  
\begin{equation}
E({\bf {x}}^{N}) \equiv \Psi^{-1}({\bf {x}}^{N}) \hat{H} \Psi({\bf {x}}^{N})  
\end{equation}
and $p({\bf {x}}^{N})$ is the probability density of the system in configuration $\bf {x}^{N}$
\begin{equation}
p({\bf {x}}^{N}) = \frac{ \left| \Psi({\bf {x}}^{N})\right|^2 }{ \int d{\bf {x}}'^{N} \left| \Psi({\bf {x}}'^{N})\right|^2} 
\end{equation}

The Metropolis algorithm provides a convenient way to sample the configurational space with probability $p({\bf {x}}^{N})$. The probability of acceptance for transition from configuration ${\bf {x}}_{n}^{N}$ to ${\bf {x}}_{o}^{N}$ is given by   
\begin{equation}
\label{eq:MA}
acc( {{\bf {x}}_{n}^{N} | {\bf {x}}_{o}^{N} }) = min \bigg\{ 1, \frac { \tau( {\bf {x}}_{o}^{N} | {\bf {x}}_{n}^{N} ) |\Psi({\bf {x}}_{n}^{N})|^{2} } { \tau( {\bf {x}}_{n}^{N} | {\bf {x}}_{o}^{N}) | \Psi({\bf {x}}_{o}^{N})|^{2}) } \bigg\} 
\end{equation}
where $\tau( {\bf {x}}_{o}^{N} | {\bf {x}}_{n}^{N})$ represents the probability of a trial move from configuration ${\bf {x}}_{n}^{N}$ to ${\bf {x}}_{o}^{N}$. A simple procedure to accomplish the Monte Carlo move is by a radon displacement of the electron configuration
\begin{equation}
 {{\bf {x}}_{n}^{N}} = {{\bf {x}}_{o}^{N} +  \pmb{\xi}^{N}} \Delta
\end{equation}
where $\pmb{\xi}^N$ a $3N$-dimensional vector of uniformly distributed random numbers between $-1$ and $1$, and $\Delta>0$ represents the step length. Typically, the step length is selected such that about $50\%$ of the trial moves are accepted. 

Starting with a suitable electronic structure, the Metropolis algorithm generates new configurations that will converge to $p({\bf {x}}^{N})$ after a sufficiently large number of Monte Carlo moves. As a result, the variational energy can be obtained by averaging over these ``sampled'' configurations
\begin{equation}
\label{eq:MC}
E_{v} \approx \frac{1}{M} \sum_{i=1}^{M} E({\bf {x}}_{o}^{N}) 
\end{equation}
where $M$ denotes the number of samples. In stark contrast to Eq.\eqref{eq:VP}, Eq.\eqref{eq:MC} involves no high-dimensional integration. Because the summation is independent of the dimensionality of the wave function, the Monte Carlo method thus drastically reduces the computational cost for evaluation of the variational energy. In the statistics literature, the VMC is also known as the Metropolis algorithm, which is widely used for sampling from the posterior distribution in Markov Chain Monte Carlo methods for Bayesian inference.
 
To minimize the variational energy, one may express the wave function in the so-called Jastrow-Slater form
\begin{equation}
\Psi({\bf {x}}^{N}) = e^{J({\bf {x}}^{N})} \Phi({\bf {x}}^{N})
\end{equation}    
where $J({\bf {x}}^{N})$ is known as the Jastrow factor, and $\Phi({\bf {x}}^{N})$ is a Slater determinant (or a linear combination of Slater determinants). The Jastrow factor accounts for the electron-electron and electron-nuclear correlations that neglected in $\Phi({\bf {x}}^{N})$. The correlation effects are typically written in terms of semi-empirical functions of the particle-particle distances with the parameters obtained by minimization of the ground-state energy. The Slater determinant can be obtained from the Hartree-Fock-like low-level QM calculations.

\begin{backgroundinformation}{Uniform Electron Gas}

One primordial example for applications of QMC is to study the equilibrium properties of uniform electron gas at either the ground state 0 K or at finite temperatures. Historically, the simulation results have played an instrumental role for the formulation of the local density approximation (LDA) (see Section~\ref{subsec:DFT}). From a theoretical perspective, the properties of uniform electrons also provide a useful reference for understanding inhomogeneous electronic systems and benchmark data for theoretical developments of new DFT functionals.

Figure ~\ref{fig:RDF} presents the spin-resolved radial distribution functions (RDF) for several uniform electron gases at 0 K, Here the results calculated from VMC are compared with those from a theoretical method \cite{RN36}. Similar to its classical counterpart, RDF describes the normalized local density of electrons, $g(\bf {r})=\rho(\bf {r})/\rho_{0}$, given another electron is found at the origin. For a uniform system of isotropic particles, RDF is a function of both the bulk density and the radial distance $r=|\bf {r}|$. Because of electrostatic interactions and the Pauli exclusion principle, the RDF of a uniform electron gas also depends on the spin state as well as the bulk electron density $\rho_{0}$. In Figure ~\ref{fig:RDF}, the bulk density is expressed in terms of the reduced Wigner-Seitz radius  
\begin{equation}
\label{eq:wsr}
r_{s} = \left( \frac {3} {4 \pi \rho_{0} a^{3}_{0}} \right)^{1/3}   
\end{equation} 
where $a_{0} =5.2917721... \times 10^{-11} $m is the Bohr radius.  

\begin{figure}
\centering
\includegraphics[width=120mm]{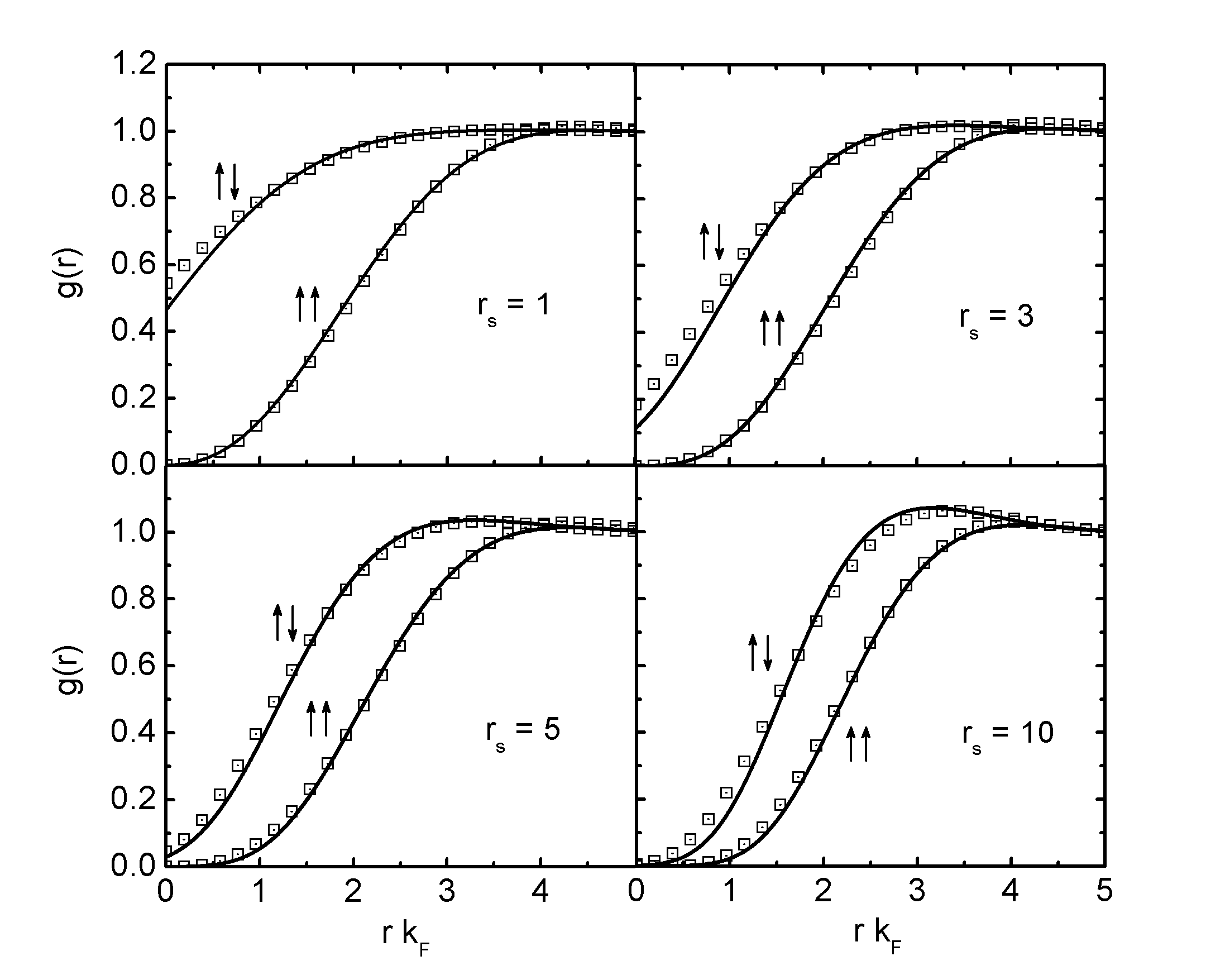}
\caption{Radial distribution functions for different electron pairs (spin up and down) in a uniform paramagnetic system at T = 0 K and different electronic densities. Here $r_{s}$ is the reduced Wigner-Seitz radius, and $k_{F}= (6\pi \rho_{0})^{1/3}$. Solid lines are from an analytical theory and symbols are from variational Monte Carlo simulation. Reproduced from  \cite{RN36}. }
\label{fig:RDF}
\end{figure}

Despite the divergence of the Coulomb potential at $r=0$, the RDF for electrons of opposite spins remains a finite value at the origin, manifesting the wave nature of electrons. The contact value falls as the reduced Wigner-Seitz radius increases from $r_{s}=1$ to 10, and approaching zero as it further increases. The density dependence suggests that the contact value of RDF arises from the electrostatic correlation, which leads to an effective attraction among electrons. At low density (\textit{e.g.}, $r_{s}=10$), the RDF exhibits the Friedel oscillation that reflects the interplay of electric repulsion and charge screening.  The same-spin electrons experience the Pauli exclusion principle thus the RDF shows a stronger depletion at short distance. Because no two electrons can be in the same quantum state, the contact value of RDF for electrons of the same spins is exactly zero. Approximately, the difference between the RDFs of the same and opposite spins reflects the so-called exchange effects.
 \end{backgroundinformation}
 
VMC represents one of many quantum Monte Carlo (QMC) simulation methods. Other popular QMC algorithms include diffusion Monte Carlo (DMC), path integral Monte Carlo (PIMC), and more recently, full configuration interaction quantum Monte Carlo (FCIQMC) \cite{RN28}. In DMC, the ground-state wave function is obtained from the stationary solution of the time-dependent Schr\"{o}dinger equation. Mathematically, the latter is equivalent to the classical diffusion equation in imaginary time, which can be represented in term of a stochastic process (a.k.a., a random walk process). Interestingly, the idea of DMC was discussed in the seminal article by Metropolis and Ulam \cite{RN30}. DMC can be used to calculate the properties of transition metal compounds, electrons at excited states, and weak intermolecular interactions. In general, it is more accurate than VMC but is also computationally much more demanding. To a certain degree, PIMC is similar to DMC but it utilizes Monte Carlo methods to sample the ``diffusion'' paths. PIMC is commonly used to study the properties of many-particles systems at finite temperature such as superfluids and plasmas. On the other hand, FCIQMC directly samples the Slater determinant with Monte Carlo methods. It is applicable to a variety of chemical systems and solids but, at present, is most suitable for relatively small systems because of the high computational cost.  

\subsection{Density functional theory}
\label{subsec:DFT} 
Since the original concepts were introduced in the mid-1960s by Pierre Hohenberg, Walter Kohn and Lu Jeu Sham \cite{RN40, RN42}, density functional theory (DFT) has evolved into one the most widely used computational tools in condensed matter physics, chemistry, materials science, and more recently, biology as well as engineering. As an alternative to conventional many-body wave function methods, DFT is drastically more efficient from a computational perspective and has been used to predict the properties of matter virtually of all kinds as reported in over ten thousand publications every year. Despite its great popularity, DFT remains one of the most misunderstood theoretical methods, not necessarily in the sense that its usefulness is questioned or that its predictions are incomprehensible due to its intrinsic connection with quantum mechanics --- which has always been mysterious, but in the sense that its foundation, limitations, and the scopes of applications or misapplications have been routinely messed up even by well-respected experts in its own field. To a certain degree, the situation is well summarized by Sean Carroll, a theoretical physicist at the California Institute of Technology, who remarked in an Op-Ed essay from New York Times  \cite{RN43}: \begin{quotation} What's surprising is that physicists seem to be O.K. with not understanding the most important theory they have.\end{quotation}
DFT had been an obscure theory and very much ignored by the scientific community for decades before it reaches today's glory. In one of his last publications \cite{RN44}, Walter Kohn wrote on the occasion celebrating fifty years of DFT: \begin{quotation}  As many theoretical chemists can confirm from personal experience, Density Functional Theory (DFT), for several decades after the publication of the Hohenberg-Kohn theorem (in 1964) was unfavorably received by many leading traditional quantum theorists of electronic structure, including John Pople.\end{quotation}
As well-known, John Pople and Walter Kohn were colleagues at the same institute for a number of years and shared the chemistry Nobel prize in 1998!  

\begin{backgroundinformation}{Basics of Statistical Mechanics}
\label{BSM}
Before discussing the generic ideas of DFT, it is instructive to recall a few basic concepts from statistical mechanics. Consider a many-body system with volume $V$, temperature $T$, and a one-body potential for each type of particles $v_{\alpha} ({\bf{r}})$. At equilibrium, the microstates constitute a grand canonical ensemble, which encompasses all quantum states of the system as described by particles in different configurations. The equilibrium properties of the system can be expressed in terms of various forms of ensemble averages  \cite{RN25}.

The one-body potential $v_{\alpha} ({\bf{r}})$ is referred as a point energy applied to each particle ${\alpha}$. This function is invariant with the system configuration, \textit{i.e.}, it is independent of the microstates of the system. For example, if we consider a uniform electron gas, the one-body potential corresponds to the negative of the electron chemical potential, $\mu$, which is a constant defined by the system temperature and the bulk electron density. For an inhomogeneous electronic system as we have discussed in Section~\ref{subsec:BO}, the one-body potential is given by that corresponding to a uniform electronic system plus the Coulomb energy due to the electron interaction with nuclei (see Eq.\eqref {eq:eex}). 

The one-body particle density is defined as an ensemble average of the number density of particle ${\alpha}$ at different microstates 
\begin{equation}
\label{eq:rho}
\rho_{\alpha}({\bf {r}}) = \left< \hat \rho_{\alpha}({\bf {r}})\right> = \sum_{i_{\alpha}} \left<\delta ({\bf {r-r}}_{i_{\alpha}}) \right> 
\end{equation}
where $\hat \rho_{\alpha}({\bf {r}})$ stands for an instantaneous particle density (\textit{e.g.}, see Eq.\eqref {eq:rho0}). In the grand canonical ensemble, the particle numbers in the system are not fixed; they fluctuate along with the microstates. 

At a given microstate, the system energy and the density profiles of all species are determined by the many-body Schr\"{o}dinger equation. The probability of the system at each microstate is then given by   
\begin{equation}
\label{eq:pnu}
p_{\nu} = \frac{1}{\Xi} {\exp \big\{-\beta[K_\nu + \Gamma_\nu + \sum_{\alpha} \int d{\bf {r}} \hat \rho_{\alpha}({\bf {r}}) v_{\alpha}({\bf {r}})] \big\}}  
\end{equation}
where $\beta=1/(k_{B}T)$, $K_\nu$ and $\Gamma_\nu$ are, respectively, the kinetic and potential energies of the system at microstate $\nu$, and $\Xi$ stands for the grand partition function 
\begin{equation}
\label{eq:gpf}
\Xi \equiv \sum_{\nu} {\exp \big\{-\beta[K_\nu + \Gamma_\nu + \sum_{\alpha} \int d{\bf {r}} \hat \rho_{\alpha}({\bf {r}}) v_{\alpha}({\bf {r}})] \big\}} 
\end{equation}
Eq.\eqref {eq:pnu} can be derived from the second law of thermodynamics \textit{i.e.}, the system entropy is maximized subject to appropriate constraints. Alternatively, it may be obtained from the Gibbs variational principle 
\begin{equation}
\label{eq:GVP}
\Omega [p_{\nu}]  \leq \Omega [p'_{\nu}]  
\end{equation}
where
\begin{equation}
\Omega [p_{\nu}]  \equiv \sum_{\nu} p_{\nu} \big[k_{B}T \ln p_{\nu} + K_\nu + \Gamma_\nu + \sum_{\alpha} \int d{\bf {r}} \hat \rho_{\alpha}({\bf {r}}) v_{\alpha}({\bf {r}})\big]
\end{equation}
and $p'_{\nu}$ stands for the probability for an arbitrary distribution of the microstates. In Eq.\eqref{eq:GVP}, the equal sign holds only when $p_{\nu}=p'_{\nu}$. 

The grand potential of the system is defined as  
\begin{equation}
\Omega \equiv -k_BT \ln\Xi = \Omega [p_{\nu}] 
\end{equation}
where $p_{\nu}$ corresponds to the equilibrium probability.  From a thermodynamic perspective, the grand potential is the free energy of an open system which takes a minimum value at equilibrium.   
 \end{backgroundinformation}
 
In a nutshell, DFT may be summarized in terms of two theorems and one corollary. These theorems were first established by Hohenberg and Kohn for inhomogeneous electronic systems at 0 K \cite{RN40} and later extended by Mermin to electronic systems at finite temperature \cite{RN71}. In essence, the Hohenberg-Kohn (HK) theorem shows a unique relationship between one-body density and one-body potential without entailing any specific knowledge of the mcirostates of a many-particle system. As a result, it holds true for electrons at 0 K as well as multi-component thermodynamic systems of either quantum or classical particles \cite{RN72, RN46,RN47,RN75}. The corollary is known as the Kohn-Sham (KS) scheme or \textit {KS ansatz}. It is instrumental for practical applications of various DFT methods for electronic systems \cite{RN42}. 

Despite its profound implications, the proof for the HK theorem (and its variations) is rather straightforward. In the following, we discuss these theorems and the corollary in the general form. 
\begin{theorem}
\label{HK1}
For a many-particle system of volume $V$ and temperature $T$, the one-body potential for each type of particles $v_{\alpha} ({\bf{r}})$, and hence all equilibrium properties of the system, can be uniquely determined by the one-body density profiles $\rho_{\alpha} ({\bf{r}})$.   
\end{theorem}

\begin{proof}
As discussed above, an open system can be defined by volume $V$, temperature $T$, the one-body potential for each type of particles  $v_{\alpha} ({\bf{r}})$. Correspondingly, there exists a set of equilibrium one-body density profiles $\rho_{\alpha} ({\bf{r}})$ corresponding to the statistical distributions of particles in the system. Suppose that  two one-body potentials, $v_{\alpha} ({\bf{r}})$  and $v'_{\alpha} ({\bf{r}})$, lead to the same one-body density, $\rho_{\alpha}({\bf {r}})$. These one-body potentials would generate two sets of probabilities for the equilibrium distributions of the microstates, $p_{\nu}$  and $p'_{\nu}$. These probabilities yield the same one-body density: 
\begin{equation}
\rho_{\alpha}({\bf {r}}) = \sum_{\nu} p_{\nu} \hat \rho_{\alpha}({\bf {r}})= \sum_{\nu} p'_{\nu} \hat \rho_{\alpha}({\bf {r}})
\end{equation}
According to the Gibbs variational principle, we have 
\begin{eqnarray}
\Omega [p_{\nu}] \leq  \sum_{\nu} p'_{\nu} \bigg[k_{B}T \ln p'_{\nu} + K_\nu + \Gamma_\nu\bigg] + \sum_{\alpha} \int d{\bf {r}} \rho_{\alpha}({\bf {r}}) v_{\alpha}({\bf {r}}) 
 \notag\\ =\Omega' [p'_{\nu}] + \sum_{\alpha} \int d{\bf {r}} \rho_{\alpha}({\bf {r}}) [v_{\alpha}({\bf {r}})-v'_{\alpha}({\bf {r}})]
\end{eqnarray}
As both $p_{\nu}$  and $p'_{\nu}$ correspond to equilibrium distributions for the microstates of the system, the same inequality holds when primed and unprimed quantities switch the positions,
\begin{equation}
\Omega' [p'_{\nu}] \leq  \Omega [p_{\nu}] + \sum_{\alpha} \int d{\bf {r}} \rho_{\alpha}({\bf {r}}) [v'_{\alpha}({\bf {r}})-v_{\alpha}({\bf {r}})]
\end{equation}
Because the particle density is everywhere non-negative, the only way to satisfy both inequalities is  $v_{\alpha} ({\bf{r}})= v'_{\alpha} ({\bf{r}})$. In other words, the one-body potentials must be uniquely determined by the one-body density profiles. 

Theorem \ref{HK1} indicates that, in principle, one can determine the one-body potentials from the one-body density profiles. With the one-body potentials, all equilibrium properties of the systems, including the distribution of microstates $p_\nu$, can be subsequently calculated by using standard statistical-mechanical methods.    
\end{proof}

\begin{theorem}
\label{HK2}
For any system of volume $V$, temperature $T$, and a one-body potential for each type of particles $v_{\alpha} ({\bf{r}})$, the equilibrium one-body density profiles $\rho_{\alpha} ({\bf{r}})$ is determined by minimizing the grand potential
\begin{equation}
\Omega [\rho_{\alpha} ({\bf{r}})]  \equiv  F[\rho_{\alpha} ({\bf{r}})] + \sum_{\alpha} \int d{\bf {r}} \rho_{\alpha}({\bf {r}}) v_{\alpha}({\bf {r}})
\end{equation}
where $F$ is known as the intrinsic Helmholtz energy 
\begin{equation}
F[\rho_{\alpha} ({\bf{r}})] \equiv  \sum_{\nu} p_{\nu} \big[k_{B}T \ln(p_{\nu})+K_\nu + \Gamma_\nu \big]. 
\end{equation} 
\end{theorem}
\begin{proof}
This theorem is also known as the HK variational principle. The proof proceeds as follows. Supposed that $\rho'_{\alpha} ({\bf{r}})$ is the equilibrium density \textit{associated} with any other one-body potential $v'_{\alpha} ({\bf{r}})$, which generates microstate probability $p'_{\nu}$ for the distribution of microstates. Because $v'_{\alpha} ({\bf{r}})$ and subsequently $p'_{\nu}$ are uniquely determined by $\rho'_{\alpha} ({\bf{r}})$, we can rewrite the Gibbs variational principle (see Eq. \eqref{eq:GVP}) as
\begin{equation}
\Omega [\rho_{\alpha} ({\bf{r}})]  \leq  \Omega [\rho'_{\alpha} ({\bf{r}})] 
\end{equation}
Therefore, the equilibrium density $\rho_{\alpha} ({\bf{r}})$ minimizes the grand potential. 

It is worth noting that $\rho'_{\alpha} ({\bf{r}})$ must be \textit{associated} with some meaningful one-body potential $v'_{\alpha} ({\bf{r}})$. Otherwise,  $p'_{\nu}$ is not even defined and thus the inequality may not be valid. The inherent constraint of the density profiles in the HK variational principle is known as \textit {``v-representable densities''} \cite{RN76}.
\end{proof} 

\begin{corollary}
\label{KS}
For any equilibrium system of volume $V$, temperature $T$, and one-body density profiles $\rho_{\alpha} ({\bf{r}})$, there exists a non-interacting reference system that reproduces the one-body density profiles.
\end{corollary}

Formally, any intrinsic Helmholtz energy can be expressed in terms of that corresponding to an non-interacting reference system of the same $V$ and $T$, $F_{0}[\rho_{\alpha} ({\bf{r}})]$, plus the difference, $\Delta F[\rho_{\alpha} ({\bf{r}})] $:
\begin{equation}
F[\rho_{\alpha} ({\bf{r}})] \equiv F_{0}[\rho_{\alpha} ({\bf{r}})] + \Delta F[\rho_{\alpha} ({\bf{r}})] 
\end{equation} 
Given a set of density profiles, $\rho_{\alpha} ({\bf{r}})$, theorem \ref{HK1} indicates that a unique set of one-body potentials, $v_{0, \alpha} ({\bf{r}})$, can be determined for the reference system. According to theorem \ref{HK2}, the density profiles minimize the grand potential of the system under consideration as well as that of the non-interacting reference system
\begin{eqnarray}
\Omega [\rho_{\alpha} ({\bf{r}})]  =F_{0}[\rho_{\alpha} ({\bf{r}})] + \Delta F[\rho_{\alpha} ({\bf{r}})] + \sum_{\alpha} \int d{\bf {r}} \rho_{\alpha}({\bf {r}}) v_{\alpha}({\bf {r}}) \\
\Omega_{0} [\rho_{\alpha} ({\bf{r}})]  =  F_{0}[\rho_{\alpha} ({\bf{r}})] + \sum_{\alpha} \int d{\bf {r}} \rho_{\alpha}({\bf {r}}) v_{0, \alpha}({\bf {r}})
\end{eqnarray}
Following the HK variational principle $\delta \Omega / \delta \rho_{\alpha} ({\bf{r}}) = \delta \Omega_{0} / \delta \rho_{\alpha} ({\bf{r}}) = 0$, we obtain an explicit expression for the one-body potentials of the reference system 
\begin{equation}
v_{0, \alpha}({\bf {r}}) = v_{\alpha}({\bf {r}}) + \frac {\delta \Delta F[\rho_{\alpha} ({\bf{r}})]} {\delta \rho_{\alpha} ({\bf{r}})} 
\end{equation} 
Because the intrinsic Helmholtz energy for a system of non-interacting particles is relatively easy to evaluate, the KS scheme provides a feasible way to carry out the HK variational principle without specific knowledge about the microstates of the real many-body system. 

It is worth noting that the Hohenberg-Kohn (HK) theorem and the Kohn-Sham (KS) scheme are valid not only for many-body systems at the ground state but, in general, for any thermodynamic systems. While the vast majority DFT calculations up-to-date are concerned only with electrons at 0 K, more applications of DFT to ``multi-component'' systems are emerging in recent years. 

\begin{backgroundinformation}{The Kohn-Sham DFT}
Consider a spin-symmetric system containing $2N$ electrons at a nondegenerate ground state, the HK theorem asserts that the ground-state energy can be obtained from the variational principle. In the KS scheme, the reference system consists of non-interacting electrons, \textit{i.e.}, ideal Fermions, which provides a basis to evaluate the variational energy. 

The wave function of ideal Fermions can be expressed in terms of the Slater determinant  
\begin{eqnarray}
\label {eqn:Slater}
	\Psi({\bf{r}}_1,{\bf{r}}_2,...,{\bf{r}}_N) &=& \frac{1}{\sqrt{N!}} 
		\left|
			\begin{array}{cccc}
				\psi_1({\bf{r}}_1) & \psi_2({\bf{r}}_1) & \ldots & \psi_N({\bf{r}}_1) \\
				\psi_1({\bf{r}}_2) & \psi_2({\bf{r}}_2) & \ldots & \psi_N({\bf{r}}_2) \\
				\vdots            & \vdots            &        & \vdots            \\
				\psi_1({\bf{r}}_N) & \psi_2({\bf{r}}_N) & \ldots & \psi_N({\bf{r}}_N) \\
			\end{array}
		\right|
\end{eqnarray}
where $\psi_i({\bf{r}}), i=1, 2, \ldots, N$ represents a single-particle wave function. The Slater determinant accounts for the Fermion exchange effect that remains between ``non-interacting'' electrons. In stark contrast to that for the electronic system, the wave function for ideal Ferminions can be decomposed as a product of 3-dimensional functions. 

In the present of one-body potential $v_{0}({\bf{r}})$, the single-particle wave functions in the Slater determinant satisfy the one-particle Schr\"{o}dinger equation   
\begin{equation}
\label{eq:KS0}
\left[-\frac{\hbar ^2}{2{m_{e}}}\nabla ^2 + v_{0}(\bf{r}) \right]\psi_i(\bf{r}) = \epsilon_i \psi_i(\bf{r})
\end{equation}
where $\epsilon_i$ represents the single-particle energy of an ideal Fermion. According to the KS scheme, the single-particle wave functions of the ideal Ferminions must satisfy the orthonomal conditions (Eq.\eqref {eq:ON} ) and reproduce the electron density of the real system 
\begin{equation}
\rho({\bf {r}}) = 2 \sum^{N}_{i=1} \left| \psi_i ({\bf{r}}) \right|^{2} 
\end{equation}
where a factor of 2 accounts for spin pairs. 

To find the one-body potential in the reference system ($v_{0}({\bf{r}})$ in Eq.\eqref{eq:KS0}), we use the HK variational principle. The energy of the reference system and that of the real system are given by, respectively,
\begin{eqnarray}
\label{eqna: EE0}
E_{0}[\rho({\bf {r}})] = K_{0}[\rho({\bf {r}})] + \int d{\bf {r}} \rho({\bf {r}}) v_{0}({\bf {r}}) \\
E[\rho({\bf {r}})] = K_{0}[\rho({\bf {r}})] + \Delta E[\rho({\bf {r}})] + \int d{\bf {r}} \rho({\bf {r}}) v({\bf {r}})
\end{eqnarray}   
where $\Delta E[\rho({\bf {r}})]$ represents the difference between the intrinsic energies of the reference and real systems. In the former case, the intrinsic energy corresponds to the kinetic energy of ideal Fermions, $K_{0}[\rho({\bf {r}})]$. Meanwhile, the intrinsic energy of the electrons includes both kinetic and potential contributions. 

Formally, $\Delta E[\rho({\bf {r}})]$  can be written as
\begin{equation}
\Delta E[\rho({\bf {r}})]  = K[\rho({\bf {r}})] -K_{0}[\rho({\bf {r}})] + \Gamma[\rho({\bf {r}})] \equiv J[\rho({\bf {r'}})]  + E_{xc}[\rho({\bf {r}})] 
\end{equation}
where $J[\rho({\bf {r}})]$  is known as the Hartree energy, which accounts for the direct electrostatic interaction among the electrons
\begin{equation}
J[\rho({\bf {r}})] =  \frac{e^2}{8\pi \varepsilon_0} \int {\bf {dr}} \int {\bf {dr'}} \frac {\rho({\bf {r}})\rho({\bf {r'}})} {\left| \bf{r} - \bf{r'}\right|} 
\end{equation}
and $E_{xc}[\rho({\bf {r}})]$, an unknown quantity, defines the exchange-correlation energy.

A comparison of the functional derivatives of the real and reference energies with respect to the electron density leads to
\begin{equation}
\label{eq:v0}
v_{0}({\bf{r}}) =  v({\bf {r}}) +  v_{c}({\bf{r}}) + v_{xc}({\bf{r}})
\end{equation}
where $v_{c}({\bf{r}})$ is the Coulomb potential 
\begin{equation}
v_{c}({\bf{r}}) \equiv \delta J[\rho({\bf {r}})]/ \delta \rho({\bf {r}}) = \frac{e^2}{4\pi \varepsilon_0} \int {\bf {dr'}} \frac {\rho({\bf {r'}})} {\left| \bf{r}_{i} - \bf{r}_{j} \right|}
\end{equation}
and $v_{xc}({\bf{r}})$ is known as the exchange-correlation potential 
\begin{equation}
v_{xc}({\bf{r}}) \equiv  {\delta E^{xc}}/{\delta \rho({\bf {r}})}
\end{equation}
Substituting Eq.\eqref{eq:v0} into \eqref{eq:KS0} leads to the celebrated KS equation  \cite{RN42}
\begin{equation}
\label{eq:KSE}
\left[-\frac{\hbar ^2}{2{m_{e}}}\nabla ^2 + v_{c}({\bf{r}}) + v_{xc}({\bf{r}}) + v({\bf{r}}) \right]\psi_i(\bf{r}) = \epsilon_i \psi_i(\bf{r})
\end{equation}

So far the theoretical procedure is exact except that $E_{xc}[\rho({\bf {r}})]$ remains unknown. Because the exchange-correlation energy is related to the difference between the energy of many electrons and that of ideal Fermions with the same one-body density, an exact expression for $E_{xc}[\rho({\bf {r}})]$ can be attained only by solving the original many-body problem. One remarkable feature of DFT is that reasonable results can be achieved even with relatively simple approximations. For example, for systems such as metals that have near uniform electron densities, a reasonable guess of the exchange-correlation energy is provided by the so-called local density approximation (LDA)
\begin{equation}
E_{xc}[\rho({\bf {r}})]  = \int {\bf {dr}} \rho({\bf {r}}) \epsilon_{xc} (\rho({\bf {r}}))  
\end{equation}     
where $\epsilon_{xc} (\rho)$ is the per electron exchange-correlation energy for a uniform electron gas of density $\rho$. As discussed above, accurate results for $\epsilon_{xc} (\rho)$ can be obtained from quantum Monte Carlo simulation. Correspondingly, the exchange-correlation potential is given by
\begin{equation}
v_{xc}({\bf{r}})  = \epsilon_{xc} ({\bf{r}}) + \rho({\bf {r}}) \frac{d \epsilon_{xc} }{d \rho}  
\end{equation} 
Understandably, LDA breaks down for systems with highly inhomogeneous electron distributions. Tremendous efforts have been devoted to the development of better approximations for the exchange-correlation energy since1980s \cite{RN39}. 

With an approximate expression for the exchange-correlation energy, the KS equations can be solved with various numerical methods  \cite{RN69}. Subsequently, the ground-state energy can be calculated from 
\begin{equation}
\label{eq:Ene}
E[\rho({\bf {r}})]  = 2 \sum^{N}_{i=1} \epsilon_i -J[\rho({\bf {r}})] - \int {\bf {dr}} \rho^{2}({\bf {r}}) \frac{d \epsilon_{xc} ({\bf{r}})}{d \rho({\bf {r}})}     
\end{equation}

It is worth noting that the KS equation applies only to the reference system of ideal Fermions, \textit {i.e.}, the non-interacting reference system. While the reference system reproduces the one-body density of the real electronic system, its total energy is NOT the same as the ground-state energy of the real system. Neither the single-particle energy levels nor the single-particle wave functions of the ideal Fermions are relevant to any physical quantities of the real electronic system under consideration. In the KS scheme, the reference system is introduced in order to avoid the direct evaluation of the many-body wave functions. Another point one should keep in mind is that, at least in its original form, the KS-DFT is concerned only with the ground-state properties of electronic systems at 0 K.    
 \end{backgroundinformation}

For systems at finite temperature, DFT may be considered as implementing thermodynamics calculations in the Hilbert space: instead of using equation of state to represent thermodynamic properties as functions of macroscopic variables, DFT calculations are based on the HK variational principle with the thermodynamic properties formulated as functionals of the one-body density profiles. Given an analytical expression for the grand potential functional, one can derive all thermodynamic properties including multi-body correlation functions \cite{RN39}. While thermodynamics offers no information on the equations of state for any macroscopic systems, the HK theorems provides little insight on how one may formulate the grand-potential functional. Like equations of state for bulk thermodynamic systems, reliable density functionals can only be derived from quantum and statistical mechanics, often entailing complicated mathematical procedures. 

The KS scheme is instrumental not only in practical implementation of DFT calculations but also for formulation of the functionals. By adopting a non-interacting system as the reference, it circumvents direct consideration of the microscope details of interacting particles thereby simplifies the physical picture and greatly reduces the computational effort. A similar approach has been commonly practiced in applied thermodynamics. Whereas the functional of real systems under consideration are typically unknown, a generic strategy may be used to derive $\Delta F$, the difference between the intrinsic Helmholtz energy of the real system and that of the non-interacting reference system. The method is known as ``adiabatic connection'' in quantum mechanics, or ``adiabate principle'' and ``thermodynamic integration" in statistical mechanics \cite{RN77}. 

For a system of electrons and nuclei, its connection with the non-interacting reference system can be in general established by scaling the Coulomb potential (see Eq.\eqref{eq:H}): 
\begin{equation}
\hat{H_{\lambda}}= \hat{H_{0}} + \lambda \hat{H_{c}}
\end{equation} 
When $\lambda=0$, $\hat{H_{\lambda}}$ corresponds to the Hamiltonian of a non-interacting reference system, and $\lambda=1$ recovers that of the real system. According to the Hellmann-Feynman theorem, the variation of the system energy with any coupling parameter in the Hamiltonian satisfies
\begin{equation}
\label{eq:HFa}
\frac {dE}{d\lambda}= \left< \frac {d\hat{H_{\lambda}}} {d\lambda} \right>_{\Psi_{\lambda}}
\end{equation} 
More explicitly,  Eq.\eqref{eq:HFa} can be written as
\begin{equation}
\label{eq:HFb}
\frac {dE}{d\lambda} = \frac{e^2}{8\pi \varepsilon_0} \sum_{\alpha} \sum_{\alpha '} \int d{\bf{r}} \int d{\bf{r'}} \hat {\rho}_{\alpha,\alpha'} ({\bf{r, r'}}|\lambda) \frac{Z_{\alpha} Z_{\alpha'}} {\left| \bf{r}_{\alpha} - \bf{r}_{\alpha'} \right|}\
\end{equation} 
where 
\begin{equation}
\hat {\rho}_{\alpha,\alpha'} ({\bf{r, r'}}|\lambda)= \left< \sum_{i_{\alpha}} \sum_{i_{\alpha '}} \delta(\bf{r}-\bf{r_{i_{\alpha}}}) \delta(\bf{r}-\bf{r_{i_{\alpha'}}}) \right>_{\Psi_{\lambda}}
\end{equation} 
According to the thermodynamic integration method, the change in the free energy due to the inter-particle potential is
\begin{eqnarray}
\label{eqn:TI}
\Delta F[\rho_{\alpha} ({\bf{r}})]   &=& \int^{1}_{0} d\lambda \left<  \frac {dE}{d\lambda} \right>_{\lambda}  \notag\\
&=& \frac{e^2}{8\pi \varepsilon_0} \int^{1}_{0} d\lambda \sum_{\alpha} \sum_{\alpha '} \int d{\bf{r}} \int d{\bf{r'}} \rho_{\alpha,\alpha'} ({\bf{r, r'}}|\lambda) \frac{Z_{\alpha} Z_{\alpha'}} {\left| \bf{r}_{\alpha} - \bf{r}_{\alpha'} \right|}\
\end{eqnarray} 
where subscript ${\lambda}$ denotes the ensemble average over the configurations of system with rescaled Hamiltonian $\hat{H_{\lambda}}$, and the two-body density function is defined as 
\begin{equation}
\rho_{\alpha,\alpha'} ({\bf{r, r'}}|\lambda)=\left< \hat {\rho}_{\alpha,\alpha'} ({\bf{r, r'}}|\lambda) \right> 
\end{equation}

In statistical mechanics, the two-body density is often expressed in terms of the radial distribution function 
\begin{equation}
g_{\alpha,\alpha'} ({\bf{r, r'}}|\lambda) \equiv \frac{\rho_{\alpha,\alpha'} ({\bf{r, r'}}|\lambda)}{\rho_{\alpha} ({\bf{r}})\rho_{\alpha'} ({\bf{r'}})}
\end{equation}
or the total correlation function 
\begin{equation}
h_{\alpha,\alpha'} ({\bf{r, r'}}|\lambda) \equiv \frac{\rho_{\alpha,\alpha'} ({\bf{r, r'}}|\lambda)}{\rho_{\alpha} ({\bf{r}})\rho_{\alpha'} ({\bf{r'}})}-1
\end{equation}
Correspondingly, $\Delta F$ can be written as
\begin{equation}
\Delta F[\rho_{\alpha} ({\bf{r}})]   = J[\rho_{\alpha} ({\bf{r}})] + F_{xc}[\rho_{\alpha} ({\bf{r}})] 
\end{equation} 
where
\begin{eqnarray}
F_{xc}[\rho_{\alpha} ({\bf{r}})]  = \frac{e^2}{8\pi \varepsilon_0} \int^{1}_{0} d\lambda \sum_{\alpha} \sum_{\alpha '} \int d{\bf{r}} \int d{\bf{r'}} 
 \notag\\  {\rho_{\alpha} ({\bf{r}})\rho_{\alpha'} ({\bf{r'}})} h_{\alpha,\alpha'} ({\bf{r, r'}}|\lambda) \frac{Z_{\alpha} Z_{\alpha'}} {\left| \bf{r}_{\alpha} - \bf{r}_{\alpha'} \right|}. 
\end{eqnarray}

Alternatively, the exchange-correlation free energy may be written as
\begin{eqnarray}
F_{xc}[\rho_{\alpha} ({\bf{r}})]  = \frac{e^2}{8\pi \varepsilon_0} \sum_{\alpha} \sum_{\alpha '} \int d{\bf{r}} \int d{\bf{r'}} 
 {\rho_{\alpha} ({\bf{r}}) \rho^{xc}_{\alpha,\alpha'} ({\bf{r, r'}})} \frac{Z_{\alpha} Z_{\alpha'}} {\left| \bf{r}_{\alpha} - \bf{r}_{\alpha'} \right|} 
\end{eqnarray}
where the exchange-correlation hole is defined as 
\begin{equation}
\rho^{xc}_{\alpha,\alpha'} ({\bf{r, r'}}) =  \int^{1}_{0} d\lambda \rho_{\alpha'} ({\bf{r'}}) h_{\alpha,\alpha'} ({\bf{r, r'}}|\lambda) \equiv \rho_{\alpha'} ({\bf{r'}}) \bar{h}_{\alpha,\alpha'} ({\bf{r, r'}}; \bar{\rho}) 
\end{equation} 
where $\bar{h}$ and $\bar{\rho}$ represent some averaged quantities. Because the Hartree energy accounts for direct Coulomb energy for electrostatic interactions, the exchange-correlation hole may be understood as the Coulomb energy of a charged particle $\alpha$ with a cavity of particle  $\alpha'$. It satisfies the normalization condition
\begin{equation}
 \int d{\bf{r'}}  \rho^{xc}_{\alpha,\alpha'} ({\bf{r, r'}}) =  1.
\end{equation} 

\begin{figure}
\centering
\includegraphics[width=120mm]{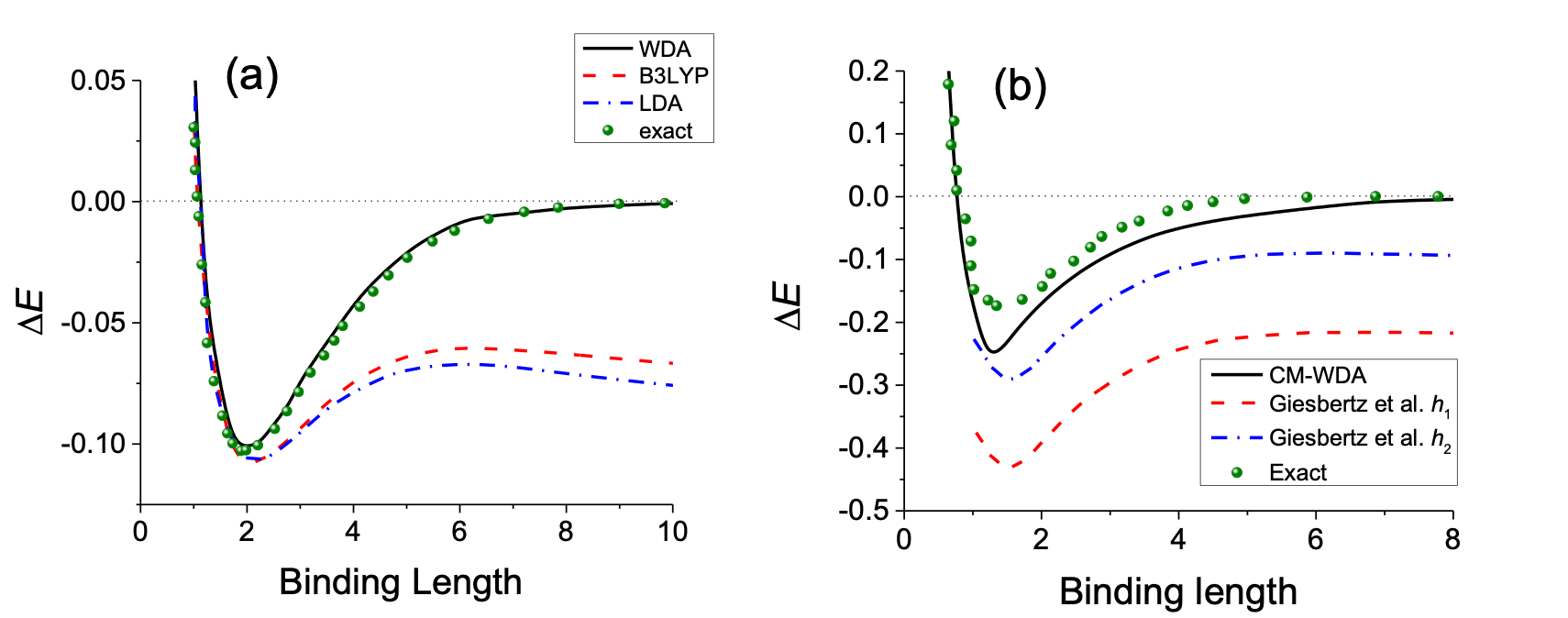}
\caption{(a) The binding energy curves (atom units) for H2+ calculated from different versions of KSDFT. B3LYP, LDA and exact results are from reference. The binding length is referred to the center-to-center distance between the two H atoms. (b) Binding energy curves for H2 calculated from different methods (atom units). Reproduced from \cite{RN79}. }
\label{fig:H2}
\end{figure}

Although an analytical expression for the exchange-correlation free energy is difficult to attain, the exact equations from the adiabatic connection are appealing because the physical meanings of various correlation function are rather intuitive. For example, Figure ~\ref{fig:H2} shows various DFT predictions for the binding energy curves for $H_{2}^{+}$ and $H_{2}$  \cite{RN79}. The solid lines are predictions from the adiabatic connection with the total correlation function represented by a simple weighted density approximation (WDA)
\begin{equation}
\bar{h} ({\bf{r, r'}}|\lambda) \approx h^{UEG} ({\left| \bf{r} - \bf{r} \right|},\bar{\rho})
\end{equation}    
where superscript ``UEG'' stands for uniform electron gas.  Whereas noticeable discrepancies are observed in comparison with exact results, WDA is free of delocalization (\textit {viz.}, no self-interaction in the single electron limit)and static correlation errors (\textit {viz.}, no binding energy between atoms in large separation) that are commonplace in many popular DFT functionals \cite{RN58}.       

Like many differential equations derived from physical models, the Schr\"{o}dinger equation is lack of an analytical solution with closed-form expressions. Conventionally, these equations are solved with the Galerkin methods, i.e., discretization of the partial differential equations into algebraic equations such that they become suitable for efficient implementation on a computer. Both plane-wave formalism and real space methods are well advanced for solving the Schr\"{o}dinger equation (and related DFT methods) \cite{RN80}.  In general, the numerical method have high computational complexity, which limits their applications to large systems of practical interests. Complementary to the numerical methods for solving the differential equations directly, the statistical and machine-learning models have long been utilized to emulate the numerical results and speed up theoretical predictions  \cite{RN81, RN82}. In the next section, we outline some recent developments in statistical and machine learning methods that offer an alternative way to circumvent solving computationally expensive the partially different equations directly. 
 
 \section{Gaussian process for scalar-valued functions}
\label{sec:GP}
 
  Gaussian process (GP) is {a} large class of statistical models that offer an alternative way to emulate a computationally expensive function with drastically less computational cost, and at the same time, has an internal assessment of uncertainty in emulation. Under some regularity conditions,  the estimator of the GP regression guarantees to converges to the true underlying function with respect to certain metric (e.g. $L_{\infty}$ or $L_2$ distance), with a known convergence rate as a function of the number of observations and ``smoothness" of truth.

  GP has been  widely used  for approximating computationally expensive computer models  \citep{sacks1989design,bayarri2007framework,higdon2008computer,spiller2014automating}. 
The statistical framework of {a} GP emulator is closely connected to the reproducing kernel Hilbert space and the kernel ridge regression (KRR), though  GP and KRR seem to be independently developed from two streams of research communities. In this section, we first introduce GP emulation and GP regression for scalar-valued functions from the probabilistic point of view in Section \ref{subsec:gpe} and Section \ref{subsec:gpr}, respectively.  The mathematical connection to the KRR and reproducing kernel Hilbert space is introduced in Section \ref{subsec:krr}, and the convergence properties that underpin these methods are introduced in Section \ref{subsec:convergence}. In the context of first-principles calculations,  GP models with new descriptors and kernels are developed  to maintain various physical properties, such as translational, permutational and rotational invariant properties \cite{bartok2013representing}. The recent advances of GP models for reproducing macroscopic quantities such as energy and mechanical properties, as well as atomic forces  for MD simulations, will be introduced in Section \ref{subsec:application_GP_scalar}.




\subsection{Gaussian process emulation}
\label{subsec:gpe}
Suppose we want to emulate a real-valued function with a scalar output $f_0: \mathcal X \to \mathbb R$ and $p$-dimensional input  $\mathbf x \in \mathcal X$. We model $f_0$ by a Gaussian process,  denoted as $f(\cdot) \sim \mbox{GP}(m(\cdot), \, K(\cdot,\cdot) )$, with mean $m: \mathcal X \to \mathbb R$, and a  covariance function (or a positive semidefinite kernel) $K: \mathcal X \times \mathcal X \to  \mathbb R$. Conditional on the mean and covariance function,  any marginal distribution $\mathbf f=(f(\mathbf x_1),...,f(\mathbf x_n))^T$ at $n$ inputs $\mathbf X=\{ \mathbf x_1,...,\mathbf x_n\}$ follows a multivariate normal distribution: 
\[\left((f(\mathbf x_1),...,f(\mathbf x_n))^T  \mid \mathbf m_X, K_{XX} \right) \sim \mathcal{MN} ( \mathbf m_X, \mathbf K_{XX} ) \]
where $\mathbf m_X=(m(\mathbf x_1),..., m(\mathbf x_n))^T$ is a vector of the mean and $\mathbf K_{XX}$ is an $n\times n$ covariance  matrix with the $(i,j)$th term being $K(\mathbf x_i, \mathbf x_j)$.

The mean is often modeled through a linear model of the basis functions: 
\begin{equation}
m(\mathbf x)= \mathbf h(\mathbf x) \bm \theta=\sum^q_{t=1} h_t(\mathbf x)\theta_t,
\label{equ:mean_basis}
\end{equation} 
where $h(\mathbf x)=(h_1(\mathbf x),...,h_q(\mathbf x))$ is a set of basis functions of $q$ dimensions, and  $\bm \theta=(\theta_1,...,\theta_q)^T$ is a vector of trend parameters, estimated from the data.  The mean is often held fixed to be zero in applications for simplicity, whereas a physical model of the basis functions may improve the predictive accuracy if the trend of underlying function can be captured by the basis functions. 

The covariance function (or kernel) $K(\cdot,\cdot)$ is the most critical component in a GP model. The GP is often assumed to be \textit{stationary} (or \textit{shift-invariant}), meaning that for any two input $\mathbf x=(x_{1},...,x_{p})$ and $\mathbf x'=(x'_{1},...,x'_{p})$, $K(\mathbf x,\mathbf x')=\sigma^2 c(\mathbf x-\mathbf x')$ 
 with $\sigma^2$ being a variance parameter and $c(\cdot)$ is a correlation function with $c(\mathbf 0)=1$. In modeling spatially correlated data, covariance function is often assumed to be \textit{isotropic}, where $K(\mathbf x,\mathbf x')=\sigma^2 c(||\mathbf x-\mathbf x'||)$, with $||\cdot ||$ being  the Euclidean distance. Frequently used correlation function include power exponential correlation and Mat{\'e}rn corrlation \cite{rasmussen2006gaussian}. The power exponential correlation function follows 
\begin{equation}
c(\mathbf x, \mathbf x')= \exp \left\{-\left(\frac{||\mathbf x-\mathbf x' ||}{\gamma }\right)^{\alpha}\right\}, 
\label{equ:pow_exp}
\end{equation}
with a range parameter $\gamma\in (0,+\infty)$ and roughness parameter $\alpha \in (0,2]$ 
When  $\alpha=2$, the kernel becomes the Gaussian kernel, where the sample path is infinitely differentiable. The roughness parameter of the kernel is often held fixed based on the smoothness of the process, and the range parameters are estimated from the data.  

The Mat{\'e}rn kernel follows 
\begin{equation}
c(\mathbf x, \mathbf x')=\frac{1}{2^{\alpha -1}\Gamma (\alpha )} \left( \frac{||\mathbf x- \mathbf x' ||}{\gamma } \right) ^{\alpha } \mathcal{K}_{\alpha } \left( \frac{||\mathbf x- \mathbf x' ||}{\gamma } \right), 
\end{equation}
where  $\mathcal{K}_{\alpha }$ is the modified Bessel function of the second kind with roughness parameter $\alpha$ and range parameter $\gamma$. The Mat{\'e}rn kernel has a closed-form expression when $\alpha=2M+1$ for $M\in \mathbb N$, and the sample path of the GP with Mat{\'e}rn kernel is $\lceil \alpha \rceil-1$ differentiable. When $\alpha=5/2$, for instance, the Mat{\'e}rn kernel follows: 
\begin{equation}
 c(\mathbf x, \mathbf x')=\left(1+\frac{\sqrt{5}||\mathbf x-\mathbf x' || }{\gamma}+ \frac{5||\mathbf x-\mathbf x' ||^2 }{3\gamma^2} \right)\exp\left( -\frac{\sqrt{5}||\mathbf x-\mathbf x' || }{\gamma}\right).
 \label{equ:matern_5_2}
 \end{equation}

Note that as each coordinate input of the computer model may have different scales, the stationary kernel is not flexible. A widely used anisotropic kernel  is the product kernel \citep{sacks1989design,Bayarri09,higdon2008computer}:
\[ K(\mathbf x,\mathbf x')=\sigma^2 \prod^{p}_{l=1}c_l(| x_{al}-x_{bl}|),\]
where $c_l(\cdot)$ is a correlation function of the output induced by the $l$th coordinate of the input. In the above expression, $c_l(\cdot)$  can be chosen as a power exponential correlation, Mat{\'e}rn correlation, or any other suitable correlation function. Note that the parameters in $c_l(\cdot)$  (such as the range parameter $\gamma_l$)  can be different for each $l$, and these parameters can be estimated by the maximum likelihood type of estimator \citep{bayarri2007framework,gu2018robust}, inducing a more flexible way to parameterize the correlation.

{\bf Maximum likelihood estimator}. The  process of computer model emulation often begins by selecting a set of inputs $\mathbf X=\{\mathbf x_1,...,\mathbf x_n\}$ from a space-filling design, such that the design points can evenly fill the input domain. Widely used random space filling designs include the Latin hypercube design and its extensions \citep{santner2003design},  Then we run simulator at these design inputs and obtain a set of numerical solutions, denoted as $\mathbf f_0=(f_0(\mathbf x_1),...,f_0(\mathbf x_n))^T$. These data will be used to estimated the model parameters, including the mean, variance and range parameters  $(\bm \theta, \sigma^2, \bm \gamma )$. Denote the mean basis $\mathbf H_X=(\mathbf h(\mathbf x_1)^T,..., \mathbf h(\mathbf x_n)^T)^T$, a $n\times q$ matrix of the basis functions.  
Differentiating the likelihood function with respect to the mean and variance parameters, we have a closed form expression of the maximum likelihood estimator (MLE) of mean and variance parameters:
\begin{align}
\hat {\bm \theta}&=\left(\mathbf H^T_X \mathbf {C}_{XX} \mathbf H_X\right)^{-1}\mathbf H^T_X \mathbf {C}^{-1}_{XX} \mathbf f_0 \label{equ:theta_hat} \\
\hat \sigma^2&=S^2_X/n \label{equ:sigma_2_hat}
\end{align}
 with $S^2_X=(\mathbf f_0- \mathbf H_X \hat{\bm \theta} )^T \mathbf {C}^{-1}_{XX}(\mathbf f_0- \mathbf H_X \hat{\bm \theta} )$
 and $\mathbf C_{XX}=\mathbf K_{XX}/\sigma^2$ being a correlation matrix with the diagonal entry being 1. 
Plugging the MLE of the mean and variance parameters into the likelihood function leads to profile likelihood of the range parameters in the kernel:
\begin{equation}
\mathcal L(\bm \gamma )\propto |{\mathbf C}_{XX}|^{-\frac{1}{2}} (S^2_X)^{-\frac{n}{2}}. 
\label{equ:profile_gamma}
\end{equation}
The range parameters $\bm \gamma$ are often estimated by numerically maximizing the natural logarithm of Equation (\ref{equ:profile_gamma}) based on a Newton algorithm \citep{nocedal1980updating}, since the closed form MLE expression may not exist. 

The MLE  is an efficient estimator of the parameters when the sample is large. When the number of available runs of a computer experiment is small, however, the MLE of the parameters of a GP emulator can be very unstable. Other estimators, such as the penalized MLE \citep{li2005analysis} and robust marginal posterior mode estimator \citep{gu2018robust}, were studied when the sample size is small. Besides, the predictive mean in equation (\ref{equ:m_star}) may be used to estimate the range parameters through cross-validation. However,  more runs may be needed than the MLE, as one needs to split the observations to estimate the parameters in a cross-validation approach.


{\bf Predictive distribution}.  
Suppose we are interested in predicting the model value at input $\mathbf x$ not run before. The joint distribution follows a multivariate normal distribution
    \[
 \left( {\begin{array}{*{20}{c}}
   f(\mathbf x) \\
   \mathbf f  \\
 \end{array} } \right)  \mid  \bm{ \hat \theta}, \hat \sigma^2, \bm { \hat \gamma} \sim {\mathcal{MN} } \left( \left(\begin{array}{*{20}{c}}
 \mathbf h(\mathbf x) \bm{\hat \theta} \\
   \mathbf H_X \bm{\hat \theta}  \\
 \end{array}
   \right), \,
  \left( {\begin{array}{*{20}{c}}
  K(\mathbf x,\mathbf x) &\mathbf K_{xX}  \\
 \mathbf K_{Xx}  & \mathbf K_{XX} \\
 \end{array} } \right) \right),
 \]
where 
$\mathbf  K_{xX}=\mathbf  K^T_{Xx}=(K(\mathbf x,  \mathbf x_1),...,K(\mathbf x_n,  \mathbf x_n))$, with the variance and range parameters plugged into the kernel function $K(\cdot,\cdot)$. After obtaining the observations $\mathbf f=\mathbf f_0=(f_0(\mathbf x_1),...,f_0(\mathbf x_n))^T$, by the conditional distribution of the multivariate normal,  the predictive distribution of the Gaussian process at any input $\mathbf x$ follows a normal distribution:
\begin{equation}
 \left(f(\mathbf x)\mid \mathbf f_0,  \bm{ \hat \theta}, \hat \sigma^2, \bm { \hat \gamma}\right) \sim \mathcal N(m^*(\mathbf x), \, K^*(\mathbf x, \mathbf x) ),
 \label{eq:pred_dist}
 \end{equation}
where predictive mean and covariance follows 
\begin{align}
m^*(\mathbf x)&= \mathbf h(\mathbf x) \bm{\hat \theta}+ \mathbf K_{xX}\mathbf K^{-1}_{XX}(\mathbf f_0- \mathbf H_X \bm{\hat \theta}  ) \label{equ:m_star} \\
K^*(\mathbf x, \mathbf x)&=K(\mathbf x, \mathbf x)-\mathbf K_{xX}\mathbf K^{-1}_{XX}\mathbf K_{Xx}. \label{equ:K_star}
\end{align}
The predictive mean in (\ref{eq:pred_dist}) is often used as a point estimator for predicting the value of the function {at} any $\mathbf x \in \mathcal X$.  The GP emulator has an internal assessment of the uncertainty, as the predictive variance and any quantile of the prediction can be computed by (\ref{eq:pred_dist}).

{\bf Interpolator}.  Note that if $\mathbf x=\mathbf x_i$, for any $i=1,2,...,n$, we have $\mathbf K_{xX}\mathbf K^{-1}_{XX}=\mathbf e^T_i$, where $\mathbf e_i$ is a vector with $1$ at the $i$th entry and 0 at other entry.
The predictive mean $m^*(\cdot)$ in (\ref{eq:pred_dist}) is an \textit{interpolator}, as if $\mathbf x=\mathbf x_i$, for any $i=1,2,...,n$, the predictive mean is exactly the same as  $f_0(\mathbf x_i)$:
\[m^*_f(\mathbf x)= h(\mathbf x_i) \bm{\hat \theta}+ \mathbf e^T_i (\mathbf f_0- \mathbf H_{X} \bm{\hat \theta} )=f_0(\mathbf x_i).\]
An interpolator is typically suitable when the computer model is deterministic and the numerical error from the computer model is very small.

\begin{figure}[t]
\centering
\begin{tabular}{cc}
\includegraphics[height=.4\textwidth,width=.5\textwidth ]{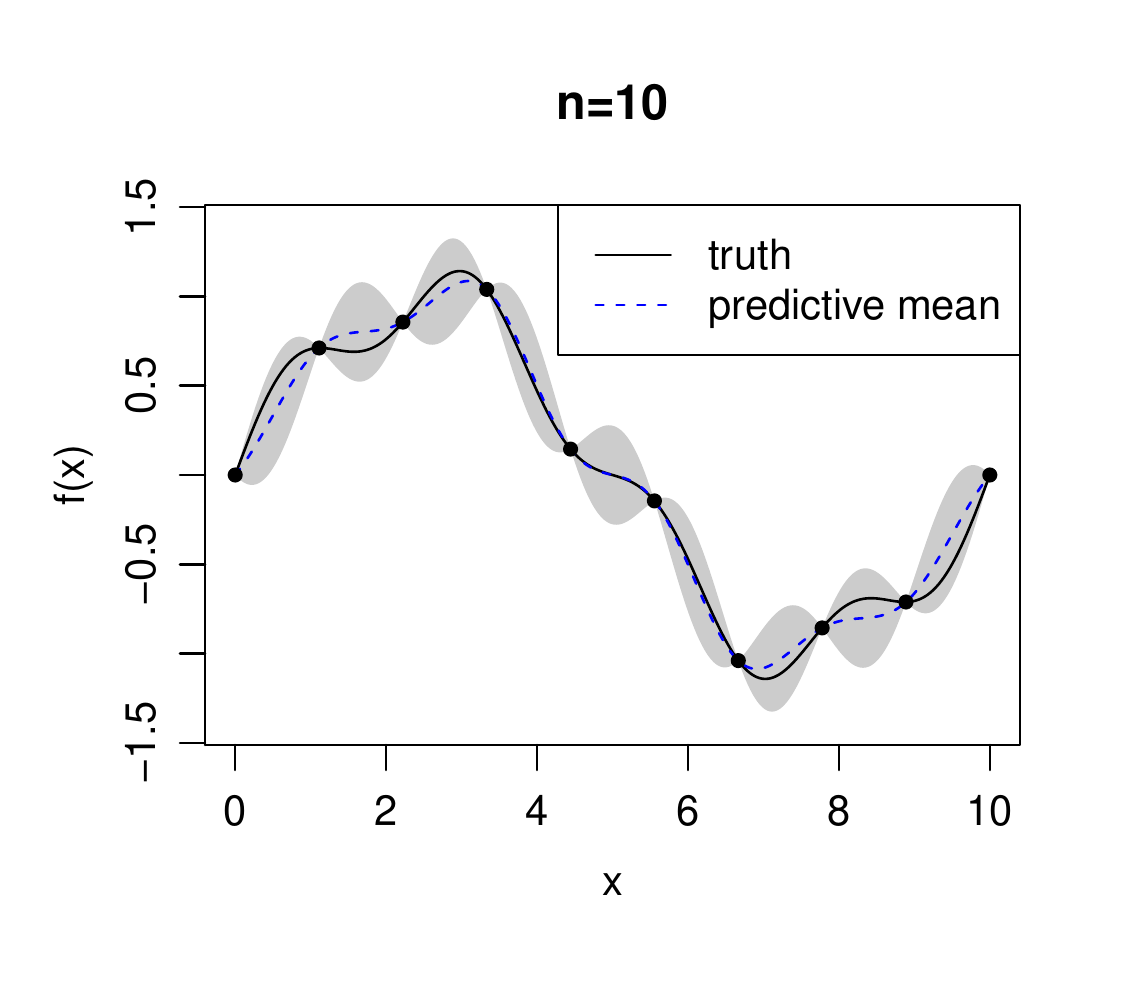}
\includegraphics[height=.4\textwidth,width=.5\textwidth]{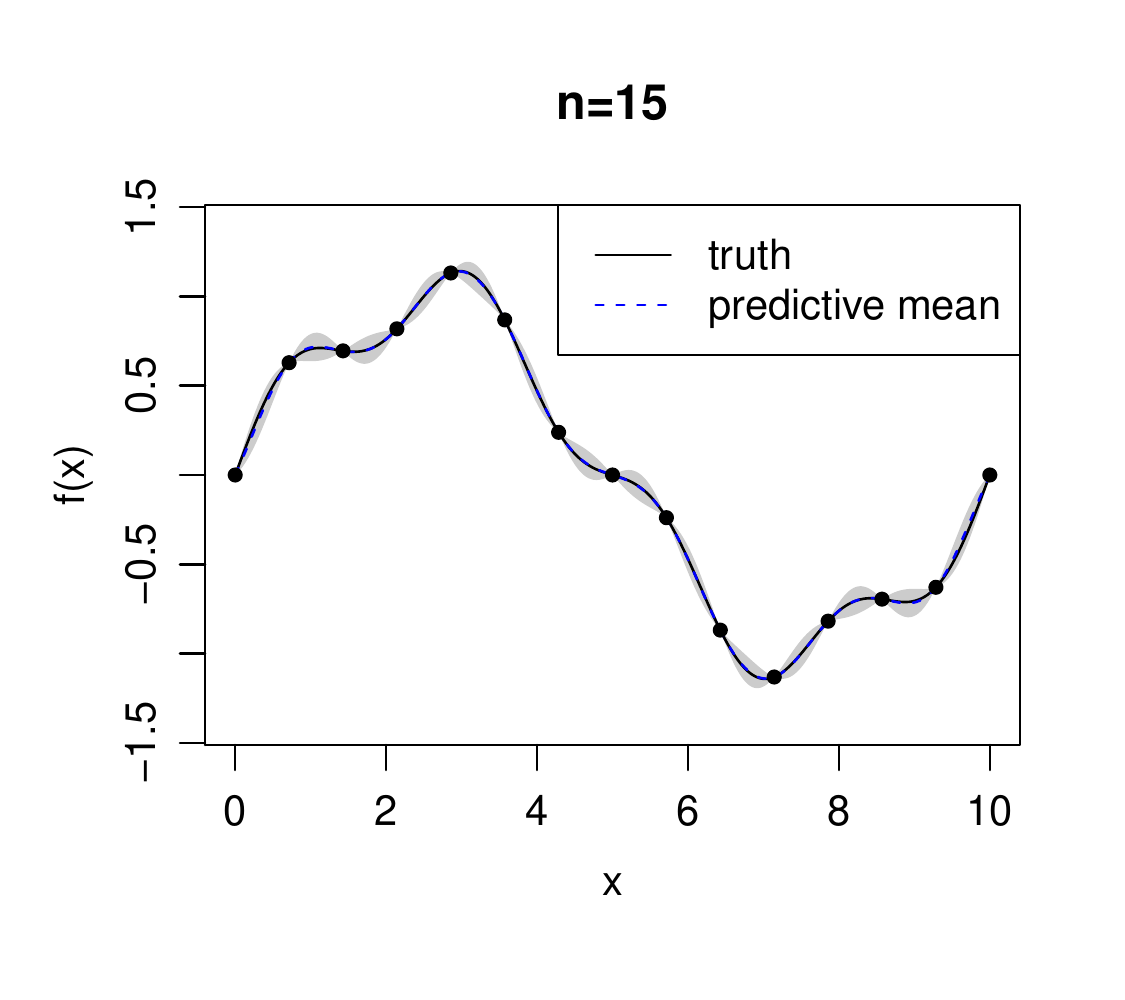} \vspace{-.2in}
\end{tabular}
\caption{Emulation of function $f(x)=\mbox{sin}(2 \pi x/10) + \mbox{sin}(2 \pi x/2.5)/5$  with the number of observations being $n=10$ and $n=15$, graphed in the left panel and the right panel, respectively. The black curves are the truth and the blue curves are the predictive mean of the GP emulator in equation \ref{equ:m_star} based on the observations graphed as the black dots. The grey area is the $95\%$ predictive interval by the GP emulator. }
\label{fig:prediction_1d}
\end{figure}

In Figure \ref{fig:prediction_1d}, we graph the predictive mean and $95\%$
 predictive interval of a GP emulator implemented in {\sf RobustGaSP} {\tt R} package \cite{gu2018robustgasp}  for function $f(x)=\mbox{sin}(2 \pi x/10) + \mbox{sin}(2 \pi x/2.5)/5$ with equal-spaced design at $x\in [0,1]$. We use the Mat{\'e}rn kernel in (\ref{equ:matern_5_2}) to parameterize the covariance and the MLE for estimating parameters. When the sample size increases, the estimation becomes more accurate, and the uncertainty (shown as the shaded area) is smaller. We only show an example with only input being 1 dimensional here, whereas the GP model implemented {\sf RobustGaSP} package is applicable for multi-dimensional input and output with both noise-free or noisy observations. 

\subsection{Gaussian process regression}
\label{subsec:gpr}
When the observations of the computer model 
contain noise (e.g. by non-negligible numerical error from the computer model), one can model the observations by 
\begin{equation}
y(\mathbf x)=f(\mathbf x)+\epsilon,
\label{equ:GP_noise}
\end{equation}
where  $f(\cdot) \sim \mbox{GP}(m(\cdot), \, K(\cdot,\cdot) )$, and $\epsilon$ is an  independent Gaussian noise with variance $\sigma^2_0$. 
The covariance function for the new process $y(\cdot)$ can be expressed as $ \tilde K(\mathbf x, \mathbf x')=   K(\mathbf x, \mathbf x') + \sigma^2_0\mathbbm 1_{x=x'}$, where $\mathbbm 1_{x=x'}=1$ if $\mathbf x= \mathbf x'$ otherwise $0$. The MLE  of $\bm \theta$ and $\sigma^2$ follow similarly in  (\ref{equ:theta_hat}) and  (\ref{equ:sigma_2_hat})  by replacing $\mathbf f_0$ and $C_{XX}$ by $\mathbf y$ and $\tilde C_{XX}:=C_{XX} +\eta \mathbf I_n$, respectively, where $\eta=\sigma^2_0/\sigma^2$ is referred as the nugget parameter. 


 Denote the observations $\mathbf {y}=(y(\mathbf x_1), ...,y(\mathbf x_n))^T$. The predictive distribution  at any input $\mathbf x$ with noisy observations also follows a normal distribution, $(f(\mathbf x)\mid \bm{ \hat \theta}, \hat \sigma^2, \bm { \hat \gamma} , {\hat \eta} )\sim \mathcal N(\tilde m^*(\mathbf x), \tilde K^*_y(\mathbf x, \mathbf x))$ with the predictive mean and variance below
\begin{align}
\tilde m^*(\mathbf x)&= \mathbf h(\mathbf x) \bm{ \tilde \theta} + \mathbf C_{xX}(\mathbf C_{XX}+\eta \mathbf I_n)^{-1} (\mathbf y- \mathbf H_X \bm{ \tilde \theta} ) \label{equ:m_star_y}\\
\tilde K^*_y(\mathbf x, \mathbf x)&=\hat \sigma^2 \left(C(\mathbf x, \mathbf x)-\mathbf C_{xX}(\mathbf C_{XX}+\eta \mathbf I_n)^{-1}\mathbf C_{Xx}\right). 
\label{equ:K_star_y}
\end{align}
where $\bm{ \tilde \theta}=(\mathbf H^T_X \mathbf {\tilde C}_{XX} \mathbf H_X)^{-1}\mathbf H^T_X \mathbf {\tilde C}^{-1}_{XX} \mathbf y$, $\mathbf C_{xX}=\mathbf K_{xX}/\hat \sigma^2 $ with the variance and range parameters plugged into the kernel function $K(\cdot,\cdot)$.
  

\subsection{Connection between Gaussian process regression and kernel ridge regression}
\label{subsec:krr}
{\bf Reproducing kernel Hilbert space}. 
We call $\mathcal H$ the reproducing kernel Hilbert space (RKHS) with the native norm (or RKHS norm) $||f||_{\mathcal H}=\sqrt{\langle f, f\rangle _{\mathcal H}}$, if there exists a kernel function $K: \mathcal X \times \mathcal X \to \mathbb R$, such that, 1) for any  $\mathbf x\in \mathcal X$, the function $K(\mathbf x, \mathbf x')$ as a function belongs to $ \mathcal H$, and 2) $K$ has the reproducing property: $\langle f(\cdot),K(\cdot, \mathbf x)\rangle_{\mathcal H}=f(\mathbf x)$ for any $f$ belongs to $\mathcal H$ \citep{rasmussen2006gaussian}.

For simplicity, let us consider a GP with zero mean (i.e. $m(\mathbf x)=0$ for any $\mathbf x$). The RKHS $\mathcal H$ attached to the GP with kernel $K(\cdot,\cdot)$ is the completion of the space of all functions: 
\[\mathcal H_0=\left\{f=  \sum^n_{i=1} w_i K(\mathbf x_i, \mathbf x), \quad w_1,...,w_n \in \mathbb R,\, \mathbf x_1,..., \mathbf x_n, \mathbf x \in \mathcal X, \, n\in \mathbb N\right\}
\]
with the inner product 
\[ \left\langle\sum^{n_1}_{i=1} w_i K(\mathbf x_i,  \cdot) , \sum^{n_2}_{j=1} w_j K(\mathbf x_j,  \cdot) \right\rangle_{\mathcal H}=\sum^{n_1}_{i=1}\sum^{n_2}_{j=1} w_i w_j K(\mathbf x_i, \mathbf x_j). \]
with $n_1, n_2 \in \mathbb N$.

We denote $\langle f, g \rangle_{L_2(\mathcal X)}= \int_{\mathbf x \in \mathcal X} f(\mathbf x)g(\mathbf x)d \mathbf x$ the inner product in $L_2(\mathcal X)$. The RKHS $\mathcal H$ contains all functions $f(\cdot)= \sum^\infty_{k=1} f_k \phi_k(\cdot) \in L_2(\mathcal X)$ with $f_k=\langle f, \, \phi_k \rangle_{L_2(\mathcal X)}$  and $\sum^\infty_{k=1}f^2_k/\lambda_k<\infty$. For any $g(\cdot)=\sum^\infty_{k=1} g_k \phi_k(\cdot) \in \mathcal H$ and $f(\cdot)= \sum^\infty_{k=1} f_k \phi_k(\cdot)$, the inner product can be represented as $ \langle f, g\rangle_{\mathcal H}=\sum^\infty_{k=1} f_k g_k/\lambda_k $. For more discussion on the RKHS, see Chapter 6 in \cite{rasmussen2006gaussian} and Chapter 1 of \cite{wahba1990spline}.


{\bf Kernel ridge regression}. Consider $n$ noisy observations $\mathbf y=(y_1,y_2,...,y_n)$ with $y_i=y(\mathbf x_i)$ for $i=1,2,...,n$. We are interested to estimate the mean $f(\mathbf x)=\E[y(\mathbf x)]$ of the observations for any $\mathbf x\in \mathcal X$. Denote $\mathcal H$ the RKHS attached to kernel $K(\cdot,\cdot)$.  The kernel ridge regression (KRR) solves the following optimization problem:
\begin{equation}
\hat f_n=\underset{g\in \mathcal H}{argmin} \left\{\frac{1}{n}(y_i-g(\mathbf x_i))^2+\lambda || g ||^2_{\mathcal H} \right\}
\label{equ:f_krr}
\end{equation}
where $\lambda$ is a regularization parameter typically estimated from data.

\begin{theorem}(Solution of KRR). The solution of Equation (\ref{equ:f_krr}) is unique and has the following expression:
\begin{equation}
\hat f_n(\mathbf x)= \sum^n_{i=1}\hat w_i C(\mathbf x, \mathbf x_i)= \mathbf C_{xX}(\mathbf C_{XX}+n\lambda \mathbf I_n)^{-1} \mathbf y
\label{equ:sol_krr}
\end{equation}
for any $\mathbf x \in \mathcal X$ with $\mathbf{\hat w}=(\hat w_1,...,\hat w_n)^T= (\mathbf C_{XX}+n\lambda \mathbf I_n)^{-1} \mathbf y$. 
\label{thm:sol_krr}
\end{theorem}

\begin{proof}
By the representer lemma \citep{rasmussen2006gaussian,wahba1990spline}, for any $\bm \theta \in \bm \Theta$ and $\mathbf x \in \mathcal X$, one has 
\begin{equation*}
\hat f_n(\mathbf x)=\sum^n_{i=1}w_i K(\mathbf x_i, \mathbf x),
\end{equation*}
and denote $\mathbf w=(w_1,...,w_n)^T \in \mathbb R^{n}$ the weights in the solution. Since $\langle K(\mathbf x_i,\cdot), K(\mathbf x_j,\cdot)\rangle_{\mathcal H}= K(\mathbf x_i, \mathbf x_j)$, equation (\ref{equ:f_krr}) becomes to find  $\mathbf w$ such that  
\begin{equation}
\mathbf{\hat w}=\underset{\mathbf w\in \mathbb R^{n}}{argmin} \left\{ \frac{1}{n} (\mathbf y- \mathbf R \mathbf w )^T(\mathbf y -\mathbf R \mathbf w )+ \lambda \mathbf w^T \mathbf R \mathbf w\right\}. \label{equ:quadratic_w}
\end{equation}
Differentiating (\ref{equ:quadratic_w})  with regard to $\mathbf w_{\bm \theta}$,  we have
\begin{equation}
\mathbf {\hat w}= (\mathbf R+n\lambda \mathbf I_n)^{-1}\mathbf y. 
\label{equ:w_form}
\end{equation}
\end{proof}

\begin{remark}
The solution of KRR in (\ref{equ:sol_krr}) is exactly the same as the predictive mean estimator in equation (\ref{equ:m_star_y}) when the mean function is zero (i.e. $m(\mathbf x)=0$ for any $\mathbf x \in \mathcal X$) and $\lambda=\eta/n$. 
\label{remark:connection_KRR}
\end{remark}

The KRR solves the optimization problem in (\ref{equ:f_krr}) and gives an estimator of the function. As the noise is not modeled, the uncertainty of the KRR estimator is not specified. As stated in Remark \ref{remark:connection_KRR}, the solution of the KRR is equivalent to the predictive mean of GP regression in (\ref{equ:m_star_y}). One main advantage of the GP model is the uncertainty of the estimator can be computed based on the predictive distribution. 



Note that many simulators  may be deterministic or  may contain very small numerical error. In this scenario, the observations become  $y_i=f(\mathbf x_i)$ for $i=1,2,...,n$. The solution of KRR, however, may not be suitable for these scenarios, as it is not an interpolator. Consider the following kernel ``ridgeless" problem \citep{liang2020just}:

\begin{equation}
\hat f_n=\underset{g \in \mathcal H}{argmin} || g||_{\mathcal H}, \mbox{ subject to } g(x_i)= f(x_i), \mbox{ for } i=1,2,...,n
\label{equ:krl}
\end{equation}
The solution of equation (\ref{equ:krl}) follows \citep{kanagawa2018gaussian}:
\begin{equation}
\hat f_n=\mathbf K_{xX}\mathbf K_{XX}^{-1} \mathbf f
\label{equ:sol_krl}
\end{equation}
Note that the solution in equation (\ref{equ:sol_krl}) is exactly the same as the predictive mean expression in (\ref{equ:m_star}) with mean zero $m(\mathbf x)=0$. 


\subsection{Convergence rates}
\label{subsec:convergence}
GP regression is a flexible approach to approximate nonlinear continuous functions. We briefly introduce the convergence properties of GP regression to the true underlying function.  Suppose the observations are from
\begin{equation}
y(\mathbf x)=f_0(\mathbf x)+\epsilon,
\label{equ:npr}
\end{equation}
where $f_0(\mathbf x)$ is the true deterministic function with $\mathbf x \in \mathcal X$ and $\epsilon$ is an independent noise. Let us assume we evaluate the goodness of estimation by the $L_2$ norm: $||\hat  f_n- f_0 ||_{L_2}=(\int_{\mathbf x\in \mathcal X} (f(\mathbf x)- f_0(\mathbf x))^2 d\mathbf x)^{1/2}$, where $\hat  f_n$ is the  KRR estimator in (\ref{equ:sol_krr}) (or equivalently the predictive mean estimator of GP regression in (\ref{equ:m_star_y})). 

Loosely speaking, the convergence of KRR estimator depends on three regularity conditions. First  the noise $\epsilon$ should have a tail decreasing rate not slower than the Gaussian distribution (i.e. the sub-Gaussian distribution). Second the sequences of inputs $\{\mathbf x_i\}^{\infty}_{i=1}$ should fill the space $\mathcal X $. Third the number of {small balls} needed to  cover  the functional space should not be too large.  Denote the covering number $N(r, \mathcal F, ||\cdot ||_{L_2})$ the smallest value of $N$ for the functional space $\mathcal F$ over $\mathcal X$, such that there exists a series of $L_2$ integrable functions $\{f^L_1,...,f^L_N,f^U_1,...,f^U_N\}$ with $|| f^L_i(\cdot)-f^U_i(\cdot) ||\leq r$ and  $r>0$ for $i=1,...,N$, and for each $f\in \mathcal F$, one has $f^L_i\leq f\leq f^U_i$ for certain $1\leq i\leq N$. We refer to the book of empirical process for further discussion of the covering number \cite{van2000empirical,kosorok2008introduction}. 

 For simplicity, we assume the design follows $U([0,1]^p)$, a uniform distribution at $[0,1]^p$. Further denote $\mathcal F(\rho)=\{f\in \mathcal F, |f||_{\mathcal H} \leq \rho \}$.   We are ready to state the convergence theorem, which can be inferred by Theorem 10.2 from \cite{van2000empirical}.
 
\begin{theorem}
 Suppose the data are generated from equation (\ref{equ:npr}) with $f_0\in \mathcal F$. Suppose $\mathbf x_i \sim U([0,1]^p)$ the uniform distribution with domain $[0,1]^p$, and there exists a constant $K_0$ such that $\E_{\epsilon}[\exp(K_0 \epsilon)]< \infty$. Furthermore, there exists $\tau$, such that $log(N(r, \mathcal F(\rho), || \cdot||_{L_2} ))\lesssim \rho^{\tau} r^{-\tau}$, 
 for all $r, \rho>0$. When $\lambda^{-1}=O(n^{2/(2+\tau)})$, the $L_2$ norm of the difference between the estimator {$\hat f_n$} and truth underlying function $f_0$ is stochastically bounded by $\lambda^{1/2}$: 
 \[ || \hat f_n -f_0||_{L_2}=O_p(\lambda^{1/2}).\]
 \label{thm:convergence_krr}
 \end{theorem}
 Various conditions in Theorem \ref{thm:convergence_krr} can be relaxed. For example, the design space can be trivially extended to any bounded rectangle in $\mathbb R^p$ and the distribution of the design can also be modified to have the same convergence properties.

\begin{remark}
Various kernels and functional space satisfy the conditions.  E.g. for the Mat{\'e}rn kernel,  the RKHS is equivalent to the Sobolev space. Assuming $\mathcal X=[0,1]^p$, the natural logarithm of the covering number follows \cite{edmunds2008function,tuo2015efficient}:
 \[log(N(r, \mathcal F(\rho), || \cdot||_{L_2} ))\lesssim {\rho}^{p/\alpha}r^{-p/\alpha}. \]
where $\alpha$ is the roughness parameter and the Mat{\'e}rn kernel and with a constant $K'_1$.  Then if $f_0 \in \mathcal F$, where $\mathcal F$  is the Sobolev space and $\lambda^{-1}	\asymp n^{2\alpha/(2\alpha+p)}$ with $	\asymp$ denoting the same change of magnitude in both sides with respect to the change of $n$, we have the optimal convergence rate 
\[|| \hat f_n -f_0||_{L_2} =O_p(n^{-\frac{\alpha}{2\alpha+p}}).\]
\end{remark}

\subsection{Emulator of force in density functional theory}
\label{subsec:forece_emulation}
In \cite{chmiela2017machine}, the authors introduce  {a} \textit{gradient domain
learning (GDML) model}, based on the GP emulator of the force  with a constrained kernel constructed  by the relationship between energy and force. Denote $\mathbf x$ an descriptor of a molecule of $N$ atoms with position $(\mathbf r_1,\mathbf r_2,...,\mathbf r_N)$, where $\mathbf r_i \in \mathbb R^3$ for $i=1,2...,N$. In \cite{chmiela2017machine}, the descriptor is a $N^2$-dimensional real-valued vector  $\mathbf x =Vec(\mathbf D)$, where $Vec(.)$ is a vectorization operator and $\mathbf D$ is a $N\times N$ matrix with the $(i_1,i_2)$th entry of $\mathbf D$ being $D_{i_1,i_2}=||\mathbf r_{i_1}-\mathbf r_{i_2} ||^{-1}$ if $i_1>i_2$ and $0$ if $i_1\leq i_2$. Denote $f_E(\mathbf x)$ the total energy as a function of descriptor $\mathbf x$. The energy can be modeled as a GP {emulator}, meaning that for any set of descriptors $\{\mathbf x_1,\mathbf x_2,...,\mathbf x_n\}$, we have 
\begin{equation}
    \left((f_E(\mathbf x_1),...,f_E(\mathbf x_n))^T  \mid \mathbf m_X, \mathbf K_{XX}\right) \sim\mathcal{MN}( \mathbf m_X, \mathbf K_{XX} ),
    \label{equ:f_E_gp}
\end{equation} 
where $\mathbf m_X=(m(\mathbf x_1),...,(m(\mathbf x_n))^T$ is a vector of the mean and $\mathbf K_{XX}$ is an $n\times n$ covariance  matrix of energies with the $(l_1,l_2)$th term being $K(\mathbf x_{l_1}, \mathbf x_{l_2})$ for $l_1,l_2=1,..,n$. The isotropic Mat{\'e}rn kernel with roughness parameter being 2.5 in (\ref{equ:matern_5_2}) is used in \cite{chmiela2017machine}. 

Denote $\mathbf f_F(\mathbf r_1,\mathbf r_2,...,\mathbf r_N)$ the molecular force on atoms with positions  $\{\mathbf r_1,\mathbf r_2,...,\mathbf r_N\}$.
As the force must follow the conservation of energy, we have the following expression:
\begin{equation}
\mathbf f_F(\mathbf r_1,\mathbf r_2,...,\mathbf r_N)=-\nabla  f_E(\mathbf r_1,\mathbf r_2,...,\mathbf r_N),
\label{equ:f_F_conservation}
\end{equation}
where $\mathbf f_F(\mathbf r_1,\mathbf r_2,...,\mathbf r_N)=(f_{F,1},...,f_{F,3N})^T$ is a vector of $3N$ dimensions with $f_{F,3(i-1)+j}=\partial f_E(\mathbf r_1,\mathbf r_2,...,\mathbf r_N)/\partial r_{i,j}$ for $i=1,2,...,N$ and $j=1,2,3$. 

Since the gradient operator is a linear operator, {equations (\ref{equ:f_E_gp}) and} (\ref{equ:f_F_conservation}) imply that the marginal distribution force of any $n$ sets of molecules with descriptor $\{\mathbf x_1,\mathbf x_2,...,\mathbf x_n\}$ follows 
\begin{equation} 
\left((f_F(\mathbf x_1),...,f_F(\mathbf x_n))^T  \mid \mathbf m_X, \mathbf K_{XX} \right) \sim\mathcal MN( -\nabla \mathbf m_X, \nabla \mathbf K_{XX} \nabla^T),
\end{equation}
where $\nabla \mathbf m_X$ is a mean vector of $3Nn$ dimension with the $3(l-1)N+3(i-1)+j$th term being $\partial m_X/\partial r_{l,i,j}$, {and} $r_{l,i,j}$ being the $j$th coordinate of the $i$th atom at the $l$th molecule, for $i=1,2,...,N$, $j=1,2,3$ and $l=1,2,...,n$; $ \nabla \mathbf K_{XX} \nabla^T$ is $3Nn\times 3Nn$ covariance matrix with the $(3(l-1)N+ 3(i-1)+j,3(l'-1)N+3(i'-1)+j')$th term of $\nabla \mathbf K_{XX} \nabla^T$ being $\partial^2 K(\mathbf x_l, \mathbf x_{l'})/\partial r_{l,i,j}\partial r_{l',i',j'}$ for $i,i'=1,2,...,N$, $j,j'=1,2,3$ and $l,l'=1,2,...,n$. The matrix $\nabla \mathbf K_{XX} \nabla^T$ can be calculated using matrix derivative chain rule \cite{chmiela2017machine,petersen2008matrix}.

Without the loss of generality, assume the mean function is zero, i.e. $m(\mathbf x)=0$ for any descriptor $\mathbf x$. Denote $\mathbf f_F=(f_F(\mathbf x_1)^T,...,f_F(\mathbf x_n)^T)^T$ a $3Nn$ vector of training forces at $n$ sets of molecules. For a new molecule with any descriptor $\mathbf x$,  the predictive mean of the force on the atoms of this new molecule follows
\begin{align*}
\mathbb E[\mathbf f_F(\mathbf x) \mid \mathbf f_F]&= (\nabla \mathbf K_{xX} \nabla^T)^T (\nabla \mathbf K_{XX} \nabla^T)^{-1}\mathbf f_F \\
&=\sum^{n}_{l=1}\sum^N_{i=1}\sum^3_{j=1} \omega_{l,i,j} \nabla_{x} \frac{\partial_{x_l} K(\mathbf x,\mathbf x_l)}{\partial r_{l,i,j}  },
\end{align*}
where $\nabla \mathbf K_{xX} \nabla^T$ is a $3Nn\times 1$ vector, {with} the $3(l-1)N+3(i-1)+j$th term being $\partial_{x_l} K(\mathbf x, \mathbf x_l)/\partial r_{l,i,j}$, and $\omega_{l,i,j}$ is the  $3(l-1)N+3(i-1)+j$th term of the vector $(\nabla \mathbf K_{XX} \nabla^T)^{-1}\mathbf f_F$. {And the} $N\times N$ predictive covariance follows: 
\begin{align*}
{\mathbb{COV}}[\mathbf f_F(\mathbf x) \mid \mathbf f_F]&=\nabla K_{xx} \nabla^T-(\nabla \mathbf K_{xX} \nabla^T)^T (\nabla \mathbf K_{XX} \nabla^T)^{-1}\nabla \mathbf K_{xX} \nabla^T. 
\end{align*}


\begin{figure}[t]
\centering
\begin{tabular}{cc}
\includegraphics[height=.4\textwidth,width=.5\textwidth ]{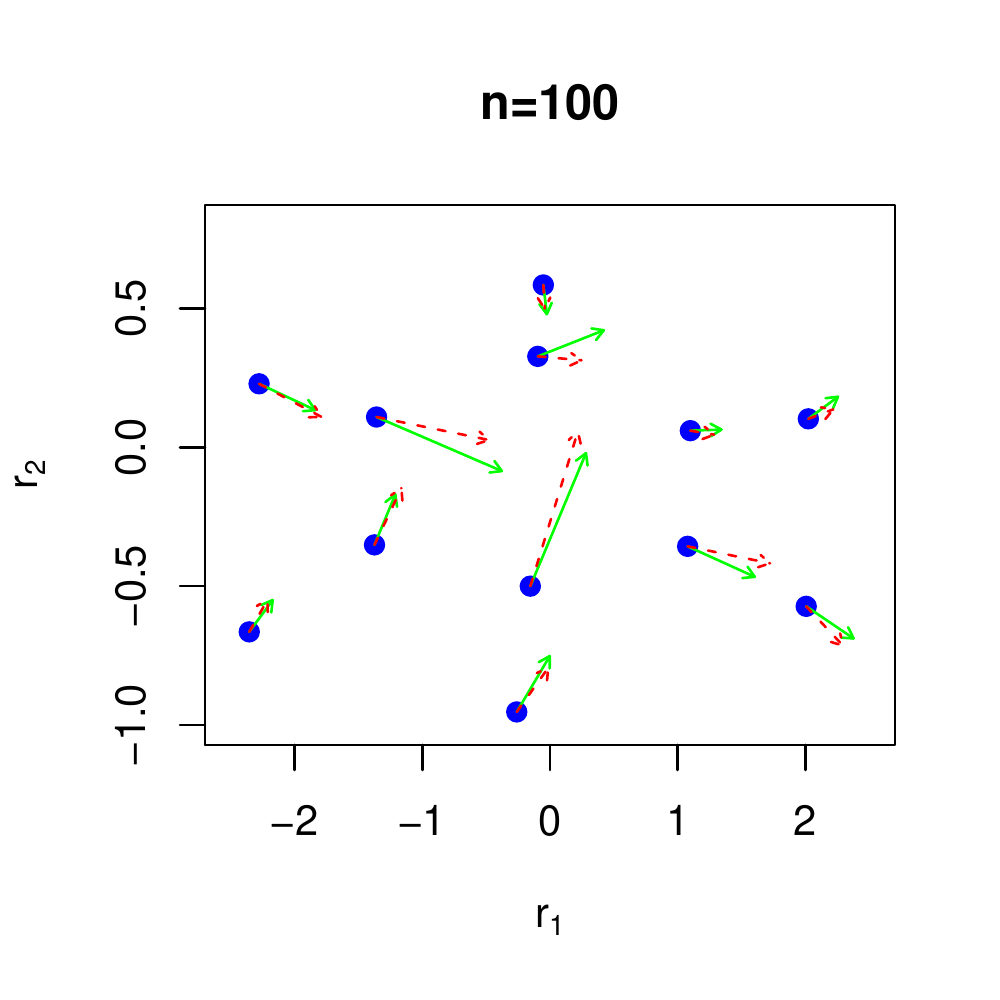}
\includegraphics[height=.4\textwidth,width=.5\textwidth]{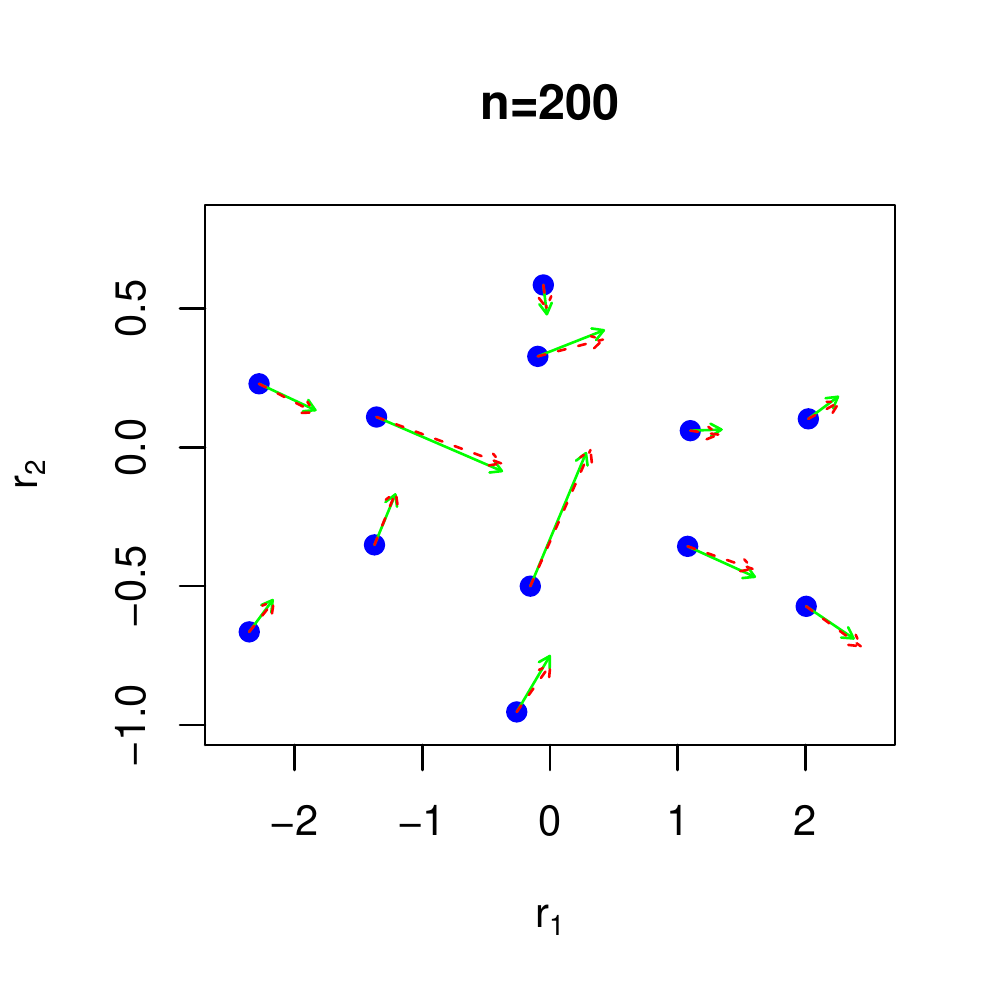} \vspace{-.2in}
\end{tabular}
\caption{Emulation of atomic force on Benzene projected on the first two dimensions, when the sample size is $n=100$ and $n=200$. The green solid arrows are the held-out test force and the red dashed  arrows are the prediction.}
\label{fig:emulation_force}
\end{figure}

Figure 4 shows the estimated force at the first two dimensions using $n=100$ and $n=200$ training data from \cite{chmiela2017machine}. When the number of observations increases, the predictions become more accurate. 

The GDML approach satisf{ies} \textit{translation} and \textit{rotation} symmetry (or invariance). The translational symmetry means the prediction of a physical quantity (such as force or energy) of any two molecules with positions $\{\mathbf r_1,...,\mathbf r_N\}$ and $\{\mathbf r_1+\mathbf h,...,\mathbf r_N+\mathbf h\}$ are the same for any real-valued vector $\mathbf h$.  The rotational symmetry means the prediction of a physical quantity remains the same when  all atoms rotate at the same angle with respect to an axis. Note that under these two operations,  the prediction of the force will not change as the descriptor of the molecule in the GDML approach remains the same.

The GDML approach does not comply with the \textit{permutation symmetry}, meaning that the  physical quantity of interest is invariant if we relabel the same atoms species. The descriptor of a molecule in the GDML approach changes after relabeling the atoms. An improved approach, called symmetrized gradient-domain machine learning (sGDML) \cite{chmiela2018towards}, seeks a permutation that minimizes the $L_2$ norm of two molecular graphs. 

For any {configuration with atomic  positions specified by} $(\mathbf r_1,...,\mathbf r_N)$, sGDML aims to find  a permutation $\tau$ to minimize the $L_2$ norm of the $N\times N$  distance matrix $\mathbf A$ with the $(i_1,i_2)$th term being $||\mathbf r_{i_1}-\mathbf r_{i_2}||$ for $i_i,i_2=1,2,...,N$. In other words, for any for two isomorphic molecular graphs with distance matrix $\mathbf A_G$ and $\mathbf A_H$, the permutation to align the matrix is estimated by $\hat \tau= \mbox{argmin}_{\tau}||\mathbf P(\tau) \mathbf A_G \mathbf P(\tau)^T-\mathbf A_H||$. The predictive performance of sGDML improves for most of the molecules considered in \cite{chmiela2018towards}.











\subsection{Emulator of energy in density function theory}
\label{subsec:application_GP_scalar}
Predicting the total energy of a system is one of the most important tasks {in first principles modeling}. In the previous GDML approach, the energy of a new molecul{ar} configuration with descriptor $\mathbf x$ can be predicted using the predictive mean $\mathbb E[f_E(\mathbf x)\mid f(\mathbf x_1),...,f(\mathbf x_n)]$ with $n$ training model runs.  Various other approaches based on the KRR estimator (or the predictive mean of Gaussian process regression) are developed in recent studies to predict the total energy of an atomic system from KS-DFT calculation.

Recent advances focus on developing new descriptors for emulating the energy. In \cite{rupp2012fast}, for instance, the descriptor of the energy of a molecule with atomic positions $(\mathbf r_1,...,\mathbf r_N)$ is specified as a pseudo Coulomb matrix $\mathbf D$, where $D_{i,j}=Z_iZ_j/||\mathbf r_i-\mathbf r_j||$ for $i\neq j$ and $D_{i,j}=0.5Z^{2.4}_i$ for $i,j=1,2,...,N$, where $Z_i$ is the nuclear 
 valence of the $i$th atom. For two molecules with descriptors $\mathbf D'$ and $\mathbf D'$, the Gaussian kernel was then used to parameterize the correlation with input $Vec(\mathbf D)$ and $Vec(\mathbf D')$. The predictive mean of the GP can be used to estimate the energy of a molecule with a new set of atom{ic} positions. 

Another recent development  is on the interatomic potential. For an atomic system, the total energy can be decomposed as \citep{bartok2015g}:
\begin{equation}
 f_E=\sum_{\alpha}\sum_{i\in \alpha}{f^{\alpha}_{E_i} } + \mbox{ long-range contributions}
 \label{equ:f_E}
 \end{equation}
where $f^{\alpha}_{E_i}$ is the local energy functionals of atom $i$ of the $\alpha$ type with compact support within a  radius $r_{cut} \in \mathbb R^+$. { The long-range contribution{s}  are referred to electrostatic, polarizability and van der Waals interactions.}


We model $f^{\alpha}_{E_i}$ as a Gaussian process  {based on inputs} representing the neighboring atomic structure. 
 A good representation of local atomic structure should be invariant to the permutational, rotational and translational symmetries, as discussed in Section \ref{subsec:forece_emulation}. Various descriptors as a function of the geometric and radial information of the neighboring atoms are developed in   \cite{bartok2013representing,bartok2010gaussian}, including power spectrum, bispectrum, radial basis and angular Fourier series. 
 After identifying the descriptor of atoms' neighboring features, the similarity between neighboring features may be measured by a kernel function.
 In \cite{bartok2018machine}, for example, the neighbor density of atom $i$ is represented as a summation of the Gaussian function 
  \begin{equation}
 \rho_i(\mathbf r)= \sum_{i'}f_{cut}(\mathbf r_{ii'}) \exp\left(-\frac{|| \mathbf  r-  \mathbf r_{ii'}||^2}{2\gamma_{atom}}\right), 
 \label{equ:rho_i_descriptor}
 \end{equation}
 where the summation is on the neighbor $i'$ atoms  including the atom $i$ itself, $\gamma_{atom}$ is a fixed range parameter and $f_{cut}(\cdot)$ is a cut-off function continuously decreases to zero beyond a cutoff radius.  The Smooth Overlap of Atomic Positions (SOAP) kernel developed in \cite{bartok2013representing} was used in \cite{bartok2018machine} to parameterize the covariance between the neighbor features of atom $i$ and $j$, denoted as $\mathbf x_i$ and $\mathbf x_j$:
 \[K( \mathbf x_i, \mathbf x_j)=\sigma^2\left|\frac{\tilde K(\mathbf x_i, \mathbf x_j)}{\sqrt{\tilde K( \mathbf x_i,  \mathbf x_i)\tilde K(\mathbf x_j, \mathbf x_j) }}\right|^{\xi}\]
 where  $\tilde K$ is defined by first integrating the square of the neighbor densities product and then  integrating over all possible 3D rotations:
 \begin{equation}
 \tilde K(\mathbf x_i, \mathbf x_j)=\int_{\mathbf{ R} \in \mbox{SO}(3)} d \mathbf { R}  \left|\int_{\mathbf r \in \mathbb R^3}  d \mathbf r  \rho_i(\mathbf r)\rho_j(\mathbf{ R} \mathbf r) \right|^{2} 
 \label{equ:K_tilde_SOAP}
 \end{equation}
 with $\mathbf R$ being a three dimensional rotation matrix in the 3D rotation group (often denoted as SO(3)), $\sigma^2$ being the variance parameter, $\xi=4$ used in \cite{bartok2018machine}. 
 As the right hand side of equation (\ref{equ:K_tilde_SOAP}) may not have a close form expression, one often requires numerical expansion.

 Assuming the total energy is normalized, so the mean {is} zero. Based on Equation (\ref{equ:f_E}) with long range correlation near zero, the covariance of the total energies $ f_{E_{a}} $ and  $f_{E_{b}}$  of two sets of atoms, denoted as $a$ and $b$,  can be computed by
 \[ \E[ f_{E_{a}} f_{E_{b}} ]= \sigma^2 \sum_{i \in \mbox{set a} } \sum_{j \in \mbox{set b}  } K(\mathbf x_i, \mathbf x_j),  \] 
as the mean is assumed to be zero. The covariance of other quantities such as force can be computed by the derivatives of the kernel functions.

The GP model with SOAP kernel achieved accurate predictive performance for  silicon clusters and the
bulk crystal  \cite{bartok2013representing,bartok2018machine}. The development of interatomic potential is ambitious, as it   allows one to use GP regression to compute the predictive distribution of the energies of possibly any set of the atoms based on the proximity of this atom set and the training atom sets. 

\section{Gaussian process emulator of vector-valued functions} 
\label{sec:GP_vector}
One important quantity in the KS-DFT {calculation} is the electron density, based on which one can compute other quantities. Various approaches are developed to emulate the electron {density}. For example, the electron density is emulated based on the Gaussian potential functions \citep{brockherde2017bypassing}. Unlike the energies and forces, the electron density is typically represented as a vector output in the Cartesian coordinate or coefficients in Fourier basis. Emulating a vector-valued function by Gaussian processes has been studies in recent year. We briefly review these approach{es} in this section. We denote the input (i.e. a descriptor function of a  set of atom{ic} positions) as $\mathbf x \in \mathbb R^{p}$ and the electron density $\bm  \rho(\mathbf x)=[\rho_1(\mathbf x),...,\rho_k(\mathbf x)]^T $ at $k$ spatial grids. 



{\bf Many single GP emulators}. The simplest approach is to model electron density at each grid independently by a GP emulator, 
$\rho_j(\cdot) \sim \mbox{GP}(m_j(\cdot),  K_j(\cdot, \cdot))$ by a mean function $m_j(\cdot)$ and covariance function $ K_{j}(\cdot,\cdot)$. The parameters of each GP emulator  can be estimated based on maximum likelihood estimator separately for each grid $j$.  The predictive distribution of any new atom set having descriptor $\mathbf x$ of any grid $j$ can be computed by the predictive distribution $(\rho_j(\mathbf x) \mid \rho_j(\mathbf x_1),...,\rho_j(\mathbf x_n), \bm \hat \beta_j, \hat \sigma^2_j, \bm {\hat \gamma}_j ) \sim \mathcal N(m^*_j(\mathbf x), \, K^*_j(\mathbf x, \mathbf x) )$, where  $(\bm \hat \beta_j, \hat \sigma^2_j, \bm {\hat \gamma}_j )$ are the estimated mean, variance and kernel parameters, respectively; $m^*_j(\cdot)$ and $K^*_j(\cdot,\cdot)$ follow equation (\ref{equ:m_star}) and (\ref{equ:K_star}), respectively, by plugging the estimated parameters $(\bm \hat \beta_j, \hat \sigma^2_j, \bm {\hat \gamma}_j )$ for grid $j$. We call this approach many single (MS) GP emulators.
 
The computational {cost} of MS GP emulators is at the order of $O(Kn^3)$, which could be  when {the number of grids or the number of training atom sets are} large. Besides, the parameters of each local GP emulator {are} estimated based on the data at each grid, which could be unstable.  

{\bf Separable GP emulator}. Noting that the output at two neighboring spatial grids is positively correlated, whereas the correlation is not exploited in the MS GP emulator. Another approach is to assume a separable GP emulator, such that $ (\rho(\mathbf x_1),...,\rho(\mathbf x_n)) \sim \mathcal{MN}(\mathbf M,  \mathbf K_{SS} \otimes \mathbf K_{XX}  )$, where $\mathbf M$ is a $k\times n$ mean matrix, $\mathbf K_{SS}$ is the covariance of spatial inputs,   $\mathbf K_{XX}$ is the covariance matrix of descriptor, and ``$\otimes$" denotes the Kronecker product. Here the covariance of data is separately modeled by a spatial covariance matrix $\mathbf K_{SS}$ and a covariance matrix of the descriptor  $\mathbf K_{XX}$. Conditional on the parameters, the predicted distribution also follows a normal distribution with mean and variance in closed-form expression \cite{wang2009bayesian}. 

When the number of spatial grids is smaller than the number of training model runs (i.e. $k<n$), a conjugate prior distribution of $ \mathbf K_{SS} $ can be specified as an inverse-Wishart distribution \citep{conti2010bayesian},  and $\mathbf K_{SS}$ can be integrated out when computing the predictive distribution. However, for a 3D electron density, the number of grids is  typically larger than the number of model runs in the training data. In this scenario, $\mathbf K_{SS}$ may be parameterized by a kernel function, where the spatial coordinate is used as the input. The computational operations of the Separable GP is $O(k^3)+O(n^3)$ in general, which is daunting for even a moderate number of grid size (e.g. $k=50^3$) for emulating the 3D electron density.  



{\bf Parallel partial GP emulator}. One computationally feasible approach is the  parallel partial Gaussian process (PP GP) emulator \cite{Gu2016PPGaSP}. In this model,  we assume the output density at grid $j$ follows $\rho_j(\cdot) \sim \mbox{GP}(m_j(\cdot),  \sigma^2_jC(\cdot, \cdot))$, which has different mean functions, different variance parameter{s} and a shared kernel function for the density at each spatial grid. Noting that {the} maximum likelihood estimator of the mean parameters and variance parameters has a closed-form expression, whereas the parameters in kernel function do not. Since the kernel function is shared across spatial grids, we only need to numerically estimate a few kernel parameters, more stable than the MS GP emulator. 

The computation complexity of PP GP emulator is $O(n^3)+O(kn^2)$ for $n$ training electron densities on  $k$ spatial grids, which is more efficient than the MS GP emulator and the separable GP emulator. The linear computational complexity with respect to $k$ allows PP GP emulator to emulate densities on a large number of grids. Even though the computational complexity of the PP GP is much smaller than the separable GP emulator, as shown in \cite{Gu2016PPGaSP}, the predictive mean of the PP GP emulator is exactly the same as the separable GP emulator, and the predictive variance between the PP GP emulator and separable GP emulator is similar. 

{\bf Semiparametric latent factor model}.  
We introduce a useful class of the linear model of coregionalization, called semiparametric latent factor model \cite{seeger2005semiparametric} for modeling the $k$-dimensional electron density at $k$ grids:
\begin{equation}
\bm  \rho(\mathbf x)=\mathbf A \mathbf z(\mathbf x)+\bm \epsilon,
\label{equ:LMC}
\end{equation}
where $ z(\mathbf x)=[z_1(\mathbf x),...,z_d(\mathbf x)]^T$  with $z_l(\cdot) \sim \mbox{GP}(m_l(\cdot), K_l(\cdot,\cdot)) $ follows a GP independently for $l=1,...,d$; $\mathbf A$ is a $k\times d$ latent factor loading matrix that relates the factor to the observations and $\bm \epsilon $ is vector of independent Gaussian noises. 

The latent factor loading matrix may be estimated by the principal component analysis, where the linear subspace is shown to be equivalent to maximum marginal likelihood estimator  (MMLE) of factor loadings when the each factor is independent, where each factor in model (\ref{equ:LMC}) follows a GP. Denote the $k\times n$ observation matrix $ {\mathbf {P}}=[\bm \rho(\mathbf x_1),...,\bm \rho(\mathbf x_n)]$ for $n$ training electron densities at $k$ spatial grids, and let  $\mathbf Z=[\mathbf z(\mathbf x_1),...,\mathbf z(\mathbf x_n)]$ be the $d\times n$ factor loading matrix. The  MMLE for latent factor loading matrix is stated in the following theorem:

	\begin{theorem}
	For model (\ref{equ:LMC}), assume $\mathbf A^T \mathbf A=\mathbf I_d$, after marginalizing out $\mathbf Z$,
	\begin{itemize}
	\item if $\bm \Sigma_1=...=\bm \Sigma_d=\bm \Sigma$. the marginal  likelihood  is maximized when  
	\begin{equation}
	 \hat {\mathbf A}=\mathbf U \mathbf R,
	\label{equ:A_est_shared_cov}
	\end{equation}
	 where $\mathbf U$ is a $k \times d$ matrix of the first $d$ principal eigenvectors of
	\begin{equation}
	\mathbf G={\mathbf  (\sigma^2_0 \bm \Sigma^{-1}+  \mathbf I_{n} )^{-1}  {\mathbf {P}}^T}, 
	\label{equ:G}
	\end{equation}
	and $\mathbf R$ is an arbitrary $d \times d$ orthogonal rotation matrix;
	\item If $\bm \Sigma_i\neq \bm \Sigma_j$ for any $i\neq j$, 
	denoting $\mathbf G_l= { {\mathbf {P}}(\sigma^2_0 \bm \Sigma^{-1}_l+\mathbf I_{n} )^{-1}{\mathbf {P}}^T}$, the maximum marginal likelihood estimator  is 
	     	\begin{equation}
\mathbf {\hat A}= \mbox{argmax}_{\mathbf A} \sum^d_{l=1}  {\mathbf a^T_l \mathbf G_l \mathbf a_l},  \quad \text{s.t.} \quad \mathbf A^T \mathbf A=\mathbf I_d,\vspace{-.2in}
	\label{equ:A_est_diff_cov}
	\end{equation}
	\end{itemize}
		\label{thm:est_A}
		\end{theorem}
	The proof of Theorem \ref{thm:est_A} along with the parameter estimation  can be found in \cite{gu2018generalized}. 
	
	
Denote $(\bm {\hat \gamma}, \bm {\hat \sigma}^2, { \hat \sigma}^2_0)$ the estimated kernel parameters, signal variance and noise variance parameters.   {Assume $\mathbf A^T \mathbf A=\mathbf I_d$, after marginalizing out $\mathbf Z$, for any $\mathbf x$, one has the predictive distribution 
 \begin{equation}
 \bm \rho(\mathbf x) \mid {\mathbf {P}}, \mathbf {\hat A},  \bm {\hat \gamma}, \bm {\hat \sigma}^2, { \hat \sigma}^2_0 \sim \mathcal{MN} \left(\bm {\hat \mu}^*(\mathbf x), \bm {\hat \Sigma}^*(\mathbf x) \right),
 \label{equ:pred_gppca}
 \end{equation}
 where the predictive mean follows
 \begin{equation}
 \bm {\hat \mu}^*(\mathbf x)=\mathbf {\hat A}  \mathbf {\hat z}(\mathbf x),
 \label{equ:hat_mu}
 \end{equation} 
 with $\mathbf {\hat z}(\mathbf x)=( {\hat z}_1(\mathbf x),...,{\hat z}_d(\mathbf x) )^T $, with ${\hat z}_l(\mathbf x)=\bm {\hat \Sigma}^T_l(\mathbf x) ({\bm {\hat \Sigma}_l+\hat \sigma^2_0 \mathbf I_n })^{-1}{\mathbf {P}}^T \mathbf {\hat a}_l$, $\bm {\hat \Sigma}_l(\mathbf x)=\hat \sigma^2_l(\hat C_l(\mathbf x_1, \mathbf x),...,\hat C_l(\mathbf x_n, \mathbf x))^T$ for $l=1,...,d$;  the predictive variance follows 
  \begin{equation}
\bm {\hat \Sigma}^*(\mathbf x)=  \mathbf {\hat A} \mathbf {\hat D}(\mathbf x) \mathbf {\hat A}^T+ \hat \sigma^2_0(\mathbf I_k - \mathbf {\hat A} \mathbf {\hat A}^T),
 \label{equ:hat_Sigma}
 \end{equation} 
 with $\mathbf {\hat D}(\mathbf x) $ being a diagonal matrix, and its $l$th diagonal term, being ${\hat D}_l(\mathbf x)=  \hat\sigma^2_l \hat C_l(\mathbf x,\, \mathbf x) + \hat \sigma^2_0 -  \bm {\hat \Sigma}^T_l(\mathbf x) \left({\bm {\hat \Sigma}_l+\hat \sigma^2_0 \mathbf I_n }\right)^{-1} \bm {\hat \Sigma}_l(\mathbf x)$,  for $l=1,...,d$.
   
  The covariance of the observations of model (\ref{equ:LMC}) follows $\sum^d_{l=1} \bm \Sigma_l\otimes \mathbf a_l\mathbf a^T_l$ where the $(i,\, j)$th term of $\bm \Sigma_l$ is $K_l(\mathbf x_i, \mathbf x_j)$ for $1\leq i,j\leq n$ and $l=1,2...,d$. Compared with separable GP emulator, this covariance matrix in LMC is not separable, representing a more flexible class of models. As we  will see in Section \ref{subsec:applications_vector}, the estimator of electron density in \cite{brockherde2017bypassing} can be written as the predictive mean in (\ref{equ:hat_mu}).
 




\subsection{Applications in emulating electron densities}
\label{subsec:applications_vector}

In  \cite{brockherde2017bypassing}, the  KRR and Fourier basis are used for emulating electron densities, based on the locations of atoms. Suppose the molecule has $N$ atoms.
The descriptor of this approach is chosen to be a function of Gaussian potential at a spatial coordinate $\bm r$:
\begin{equation}
v(\bm r)=\sum^{N}_{j=1} Z_{j} \exp\left(-\frac{|| \bm r-\bm r^{atom}_{j}||^2}{2\gamma^2} \right), 
\label{equ:v_potential}
\end{equation}
where $Z_j$ and $\mathbf r^{atom}_j$ are the nuclear charge and spatial location of the  $j$th atom, respectively; $\gamma$ is a fixed parameter. 


{Consider} $n$ observed electron densities denoted a{s} $\mathrm{\mathbf P}=[\bm \rho(\mathbf x_1),...,\bm \rho(\mathbf x_n)]$, a $k\times n$ matrix at atom{ic} configuration $[\mathbf x_1,...,\mathbf x_n]$, with $\mathbf x_i=[\mathbf r^{atom}_{i1},...,\mathbf r^{atom}_{iN}]$, where $\mathbf r^{atom}_{ij}$ is the $j$th atom{ic} position at the $i$th simulated run  for $i=1,...,n$ and $j=1,...,N$, respectively.   Denote the external potential for the electron densities at the $i$th simulated run, $i=1,...,n$, by
$\mathbf v_i=[v_i(\mathbf r_1),...,v_i(\mathbf r_k)]^T$ at locations $\{\mathbf r_1,...,\mathbf r_k\}$. Further denote the electron density of interest ${\bm \rho}(\mathbf v)=(  \rho(v(\mathbf r_1)),...,\rho(v(\mathbf r_k)))^T$ of any potential energy $\mathbf v=( v( \mathbf r_1),..., v( \mathbf r_k) )^T$ with $v(\cdot)$ following equation (\ref{equ:v_potential}). In \cite{brockherde2017bypassing},  the estimator of the electron density at any external potential $\mathbf v$ can be written as 
\begin{equation}
  \bm{\hat \rho}(\mathbf v)= \mathbf  A_r   \hat{\mathbf z}(\mathbf v),
  \label{equ:pred_krr_density}
  \end{equation}
where  $\mathbf A_{r}=[\mathbf a_1,..., \mathbf a_d]$  is a $k\times d$ basis functions over spatial coordinates  with  $\bm a_l=(a_l(\mathbf v_1),..., a_l(\mathbf v_k))^T$  for $l=1,...,d$, and the  $l$th term of $\hat{\mathbf z}(\mathbf v)=(\hat{\mathbf z}_1(\mathbf v),...,\hat{\mathbf z}_d(\mathbf v))^T$ follows
\[\hat{\mathbf z}_l(\mathbf v)=\mathbf C_{l}(\mathbf v)(\mathbf C_{l} + \lambda_{l} \mathbf I_n )^{-1} \mathbf z^{(l)},  \] 
with $\mathbf C_{l}(\mathbf v)=(C_l(\mathbf v, \mathbf v_1),..., C_l(\mathbf v, \mathbf v_n))^T$,  the $(i,j)$ term of $\mathbf C_{l}$ being $C_l(\mathbf v_i, \mathbf v_j)$ for $1\leq i,j\leq n$, $\lambda_{l}$ being a tuning parameter  and the $i$th entry of $\mathbf z^{(l)}  = ( z^{(l)}_1,...,  z^{(l)}_n)^T$ being $z^{(l)}_i=\mathbf a^T_l  \bm \rho(\mathbf x_i) $ for $l=1,...,d$ and $i=1,...,n$. 
	In \cite{brockherde2017bypassing}, the orthogonal Fourier basis (i.e. $\mathbf A^T_r\mathbf A_r=\mathbf I_d${\color {blue})} is used to parameterize the factor loading matrix and the isotropic Gaussian kernel is used to parameterize the covariance between any two electron densities with input being the external potential function. 


\begin{remark}
Suppose for any external potential $\mathbf v$, we model the electron density  by
\[\bm \rho(\mathbf v)=\mathbf A_r \mathbf z(\mathbf v), \]
where $\bm \rho(\mathbf v)$ is {the} $k\times 1$ vector of the density, and the d-dimensional factor processes $\mathbf z(\cdot)=[z_1(\cdot),...,z_d(\cdot)]^T$ are modeled as $z_l(\cdot) \sim \mbox{GP}(0,\tilde K_l(\cdot,\cdot) )$ with $\tilde K_l(\mathbf v_i,\mathbf v_j))=\sigma^2_l C_l(\mathbf v_i,\mathbf v_j))+ \sigma^2_{0l}\mathbf 1_{\mathbf v_i=\mathbf v_j}$ for $l=1,2,...,d$. The factor loading matrix $\mathbf A_r$ is  a $k\times d$ matrix with $\mathbf A^T_r\mathbf A_r=\mathbf I_d$. Conditional on the kernel and variance parameters, the estimator of the electron density in (\ref{equ:pred_krr_density}) is  equivalent to the predictive mean estimator in Equation (\ref{equ:hat_mu}) with input being the external potential. The uncertainty (e.g. $95\%$ predictive interval) can be also be obtained  by Equation (\ref{equ:pred_gppca}). 
\end{remark}

\begin{figure}[t]
\centering
\includegraphics[height=.33\textwidth,width=1\textwidth ]{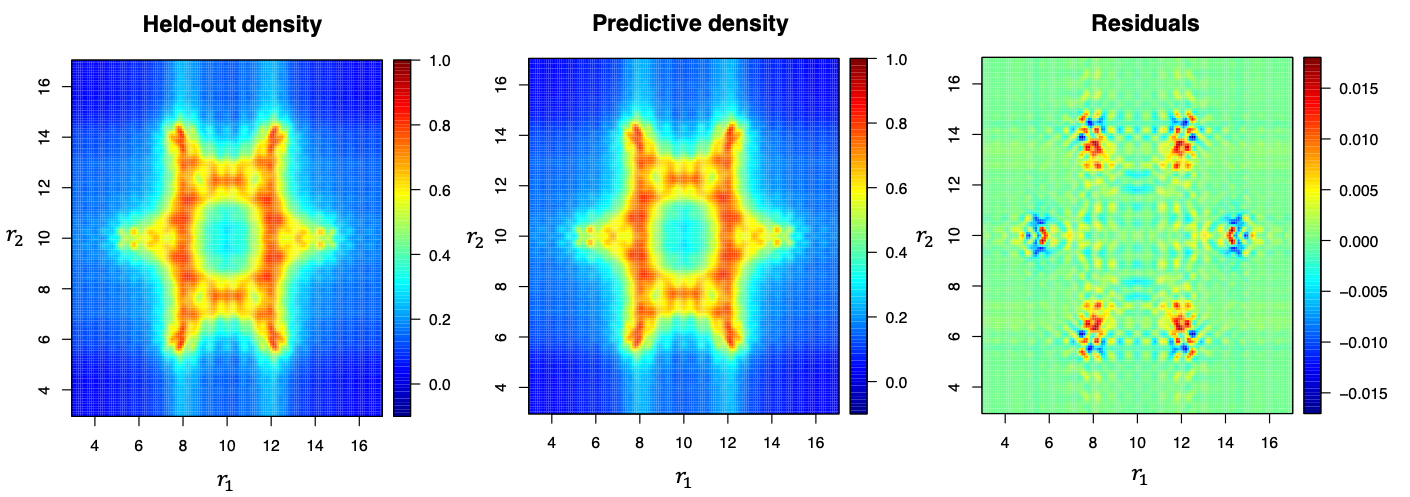}
\caption{Emulation of density at $r_3=10$ for one set of held-out density of Benzene. The held-out density, predictive density and residuals are given in the left, middle and right panel, respectively. The color scale in the right panel is  smaller than the previous two panels. }
\label{fig:emulation_density}
\end{figure}

We show the prediction of the electron density for  one set of held-output test data set in \cite{brockherde2017bypassing} based on the PP GP emulator implemented in {\sf RobustGaSP} {\tt R} \cite{gu2018robustgasp}. Only 150 electron densities are used to train our PP GP emulator, and it achieves relatively high accuracy. The  method based on the Fourier basis has similar predictive accuracy as the PP GP emulator and it is thus not shown here.






\section{Prospects}
\label{sec:Pro}
Physics-based modeling and data science are complementary but progressed almost in parallel until very recently. On the one hand, quantum and statistical mechanics calculations are able to predict the properties of virtually any matter in the universe from first principles. Whereas the theoretical foundation has been in place for nearly a century, one of the greatest dilemmas in modern science and engineering is that, owing to the seemingly insurmountable computational cost, exact equations are hardly applicable to complex systems of practical interest. On the other hand, statistical and machine learning techniques have gained a lot of attentions in recents rending new prospects for solving a wide variety of physics-based models with a tradeoff of numerical accuracy to computational cost. In conjunction with recent progress in computer technology, in particular with novel architectures such as graphical processing units (GPUs), the statistical and machine-learning algorithms will potentially overcome the major hurdles of first principles calculations.
  
A number of theoretical and simulation methods can be used for physics-based modeling. Among them, quantum Monte Carlo simulation (QMC) and the density functional theory (DFT) represent two generic theoretical frameworks that one may take to achieve accuracy and computational efficiency. In combination with machine-learning methods, these methods will potentially have transformative impacts on technological advances including the computational design of innovative devices and materials.

\begin{acknowledgement}
J.W. acknowledges financial support by the U.S. National Science Foundation’s Harnessing the Data Revolution (HDR) Big Ideas Program under Grant No. NSF 1940118.
\end{acknowledgement}

\bibliographystyle{spbasic}

\bibliography{References_2020}
\clearpage

\printauthorindex

\printindex

\end{document}